\documentclass[11pt]{article}
\usepackage[T1]{fontenc}
\usepackage[utf8]{inputenc}
\usepackage{lmodern}
\usepackage[margin=1in]{geometry}
\usepackage{amsmath,amssymb,amsthm,mathtools,bm}
\usepackage{microtype,booktabs,array,longtable,enumitem,float,multicol}
\usepackage[numbers,sort&compress]{natbib}
\usepackage[dvipsnames]{xcolor}
\usepackage[colorlinks=true,linkcolor=MidnightBlue,citecolor=MidnightBlue,urlcolor=MidnightBlue]{hyperref}
\usepackage[capitalise,noabbrev]{cleveref}
\usepackage[noend]{algorithm}
\usepackage[noend]{algpseudocode}

\hypersetup{pdftitle={The Koml\'os Problem in Matrix-Multiplication Time},pdfsubject={Algorithmic Koml\'os discrepancy, batched spectral signing, and bit complexity}}
\setlist{topsep=4pt,itemsep=3pt,parsep=1pt}
\allowdisplaybreaks[2]
\numberwithin{equation}{section}
\newtheorem{theorem}{Theorem}[section]
\newtheorem{lemma}[theorem]{Lemma}
\newtheorem{proposition}[theorem]{Proposition}

\theoremstyle{definition}

\newtheorem{remark}[theorem]{Remark}

\newfloat{algorithm}{tbp}{loa}
\floatname{algorithm}{Algorithm}
\newcommand{\R}{\mathbb R}
\newcommand{\Q}{\mathbb Q}

\newcommand{\E}{\mathbb E}
\newcommand{\eps}{\varepsilon}
\newcommand{\ones}{\mathbf 1}
\newcommand{\norm}[1]{\left\lVert #1\right\rVert}

\newcommand{\tO}{\widetilde O}
\newcommand{\op}{\mathrm{op}}

\DeclareMathOperator{\diag}{diag}
\DeclareMathOperator{\tr}{tr}
\DeclareMathOperator{\rank}{rank}

\DeclareMathOperator{\disc}{disc}

\title{The Koml\'os Problem in Matrix-Multiplication Time}
\author{Emile Anand, Peter Chen\\\\ Georgia Institute of Technology}
\date{}
\begin{document}
\maketitle
\begin{abstract}
\noindent We give a randomized algorithm that, for an $m\times n$ matrix whose
columns have Euclidean norm at most one, finds a signing with discrepancy
below $99$ in expected $\widetilde O(L_{\mathrm{in}}+n^\omega)$ bit
operations, where $L_{\mathrm{in}}$ is the input encoding length and $\omega$
is the exponent of matrix multiplication. The previous fastest algorithm with a
constant discrepancy bound, due to Li \cite{Li2026}, is deterministic and uses $O(mn+n^{\omega+2}\log^3 n)$ arithmetic
operations. We build on the spectral signing framework of Guo, Fang, and Lu \cite{GFL2026algorithm}
and of Li \cite{Li2026}, which rounds a single fractional coloring while one spectral
potential controls every row. Li's potential tracks a top eigenvalue, and the
gap below that eigenvalue, which can be of order $1/n$, limits the length of
every step. We replace it with a resolvent potential that is smooth without
any gap condition, and we draw step signs at random, so that the potential
only needs to be concave while progress is measured separately. Our key insight is to reuse a single subspace of directions throughout a batch of steps, thereby using only one spectral computation, so that the more
expensive linear algebraic operations are paid for only polylogarithmically often at each
scale. Within a batch, we sample directions from a Gaussian completion matrix
whose inverse we maintain under rank-one row replacements \cite{Brand2021}. Finally, we use a finite-precision
analysis to prove stability of the algorithm and attain the bit-complexity bound.\looseness=-1
\end{abstract}
 
\newcommand{\peter}[1]{\textcolor{red}{\textbf{Peter says:} #1}}
 
\section{Introduction}
\label{sec:introduction}
 
The Koml\'os problem asks how well we can balance vectors in the Euclidean unit ball by assigning signs of $\pm 1$ \cite{spencer1987ten}. For a matrix $\mathbf A=[a_1\ \cdots\ a_n]\in\R^{m\times n}$, its signing discrepancy is
\begin{equation}
 \disc(\mathbf A)=\min_{\eps\in\{-1,1\}^n}\norm{\mathbf A\eps}_\infty.
 \label{eq:problem}
\end{equation}
Given the promise that $\sum_{i=1}^m a_{ij}^2\le1$ for all $j\in[n]$, the Koml\'os problem asks whether an absolute constant bounds $\disc(\mathbf A)$ independently of both dimensions. In the Beck--Fiala setting \cite{BeckFiala1981}, each element of a universe belongs to at most $t$ sets, so dividing the incidence matrix by $\sqrt t$ normalizes its columns to Euclidean norm at most $1$, and a constant Koml\'os bound gives an imbalance of $O(\sqrt t)$ in every set. Vector balancing also has natural applications to rounding: the theorem of Lov\'asz, Spencer, and Vesztergombi \cite{LovaszSpencerVesztergombi1986} rounds any $x\in[0,1]^n$ to some $y\in\{0,1\}^n$ with $\|\mathbf A(y-x)\|_\infty=O(1)$, so a bound on each column's Euclidean norm controls the error in every linear measurement, independently of the number of variables and measurements. Furthermore vector balancing has seen a number of applications in machine learning and statistics such as coresets for large datasets \citep{liberty2026nearlyoptimalattentioncoresets}, attention approximation \citep{kochetkova2026streamingattentionapproximationdiscrepancy}, covariate balance in experimental design \citep{harshaw2023balancingcovariatesrandomizedexperiments}, and distribution thinning \citep{carrell2026lowrankthinning} where balancing a set of vectors in two parts allows one to approximate the whole by some carefully chosen subset. The Koml\'os problem was a longstanding conjecture in discrepancy theory until it was recently resolved by Guo, Fang, and Lu \cite{GFL2026exist}. In this work, we study the time needed to find a signing.
 
Algorithmic discrepancy theory offers several starting points. The earliest algorithms, due to Bansal \cite{Bansal2010}, used semidefinite programs to correlate the increments of a random walk and obtained algorithmic counterparts of entropy-method bounds. Lovett and Meka \cite{LovettMeka2015} replaced this optimization step by a walk in the subspace orthogonal to nearly tight constraints; their partial-coloring algorithm illustrates a principle that also guides our construction, namely to protect a small collection of constraints while the remaining dimensions make progress. For vector balancing, Banaszczyk's Gaussian-measure theorem \cite{Banaszczyk1998} gives discrepancy $O(\sqrt{\log 2n})$ under normalized column norms. Dadush, Garg, Lovett, and Nikolov \cite{DGLN2019} connected that theorem to subgaussian distributions over signed sums, which the Gram--Schmidt walk \cite{BDGL2019} samples efficiently. Guo, Fang, and Lu \cite{GFL2026exist} proved that the constant $3\sqrt{2\pi}$ suffices (see also the shorter proof of Karingula and Lovett \cite{KL2026}, with constant $36$), and introduced a polynomial-time spectral signing algorithm \cite{GFL2026algorithm}. Li \cite{Li2026} developed this framework into a deterministic algorithm with discrepancy below $99$ and arithmetic cost $O(mn+n^{\omega+2}\log^3 n)$.

 \subsection{Main Result}
For a rational-valued input, let $L_{\mathrm{in}}$ be its total binary encoding length. We count bit operations with fast integer arithmetic, and allow $\tO$ to suppress factors polynomial in $\log(2n)$. Now, we are ready to state our main result.
 
\begin{theorem}
\label{thm:main}
There exists a randomized algorithm that takes any $\mathbf A\in\Q^{m\times n}$ with column norms at most $1$, and returns $\eps\in\{-1,1\}^n$ with $\norm{\mathbf A\eps}_\infty<99$ in $\tO(L_{\mathrm{in}}+n^\omega)$ expected bit operations.
\end{theorem}
 
\begin{remark}
\label{rem:real-ram}
For real-valued inputs, the idealized algorithm uses $\tO(mn+n^\omega)$ operations in a model that assigns unit cost to field operations, comparisons, square roots, and scalar Gaussian samples, and the finite-bit algorithm approximates these normalizations and gives Gaussian digits from random bits \cite{Karney2016}. For rational entries of at most $B$ bits, the input term becomes $\tO(mnB)$.
\end{remark}

\subsection{Proof Overview}We begin by reviewing the spectral signing approach of  \cite{GFL2026algorithm,Li2026}.  Spectral signing moves a fractional coloring $x\in[-1,1]^n$ from the origin to a vertex of the cube. Every medium row carries two concave barriers, one for each sign of its discrepancy, and rows near their barriers are \emph{protected} so that, by concavity, their barriers cannot increase. The barrier weights enter an entrywise nonnegative Gram matrix $\mathbf M$, and bounding its top eigenvalue bounds $\ones^\top\mathbf M\ones$ and hence the number of protected rows. Li's algorithm \cite{Li2026} tracks this eigenvalue, slightly regularized, with a potential that is both the safety device and the progress clock: each full step chooses the better of $x\pm th$ and lowers the potential by a definite amount. That requires the top eigenvalue to be smooth along the step, hence a gap below it, and Li's regularizer guarantees a gap of only order $1/n$ at the start. Steps must therefore be tiny, and the algorithm makes $O(n^2\log n)$ updates, each paying for a fresh round of dense matrix computations. Li observes that a faster algorithm needs a potential that, like the top eigenvalue, controls the averages of the Gram matrix over every set of active coordinates, while being smooth at a coarser scale \cite[Section~12]{Li2026}. We keep Li's barriers, affine row representations, and final constant, and change three things: the potential, the way steps are signed, and how often the spectral computation is paid for. Two smaller departures are the retirement of rows with tiny mass and a rational weight function, which we describe in further detail in \Cref{sec:foundations}.

In an ambient space of $s$ active coordinates, we use the resolvent potential
\begin{equation}
 F_b(\mathbf M)=\ones^\top(b \mathbf I-\mathbf M)^{-1}\ones-\frac{s}{b},
 \label{eq:intro-potential}
\end{equation}
where $\mathbf{M}$ is PSD and entrywise nonnegative, with $\mathbf{M}\prec b\mathbf{I}$. It is a barrier for $\lambda_{\max}(\mathbf M)<b$: finite and nonnegative on that region, it blows up as $\lambda_{\max}(\mathbf M)$ increases to $b$, so along a path that starts inside the region a bounded potential keeps the top eigenvalue below $b$, and Li's count of nearly tight rows still applies. Its Neumann series has nonnegative terms, so $F_b$ decreases when entries of $\mathbf M$ decrease or a coordinate is fixed, which is the property Li asks for. Unlike the top eigenvalue, it is a rational function of $\mathbf M$ on the whole region, so repeated or crossing eigenvalues cost nothing.

Each step first exposes a \emph{proposal}, a direction $h$ and a length $t$ clipped so that $x\pm th$ both stay in the cube, and only then draws a fair sign $\sigma$. If $F_b$ is concave on the segment, then $\E_\sigma F_b(\mathbf M(x+\sigma th))\le F_b(\mathbf M(x))$, so the potential is a supermartingale and Doob's inequality keeps it below a cap. Here, we measure progress separately by the \emph{coordinate energy} $\norm{x}^2$, which grows by $t^2\norm h^2$ in expectation (\Cref{lem:walk-query-count}). The step length is therefore set, up to a logarithmic factor needed for concentration, by the inverse-polylogarithmic radius of the region where concavity is certified, and the walk takes $\tO(n)$ steps instead of Li's $O(n^2\log n)$.

However, recomputing a curvature subspace at every step would still cost $n^{\omega+1}$, and dynamic inverse maintenance does not help directly, since a step changes every row weight and hence changes $\mathbf M$ by a matrix of high rank. Therefore, we instead compute a subspace at an \emph{anchor} and reuse it for a \emph{batch} of steps. \Cref{prop:batch-geometry} finds a subspace of dimension at least $\lceil s/32\rceil$ on which the Hessian $\mathbf P-\mathbf D$ of $F_b$ satisfies $\mathbf P_0[h]\le(1-\rho)\mathbf D_0[h]$, a margin measured against the barriers' own concavity, and shows that this certificate survives a trust window of inverse-polylogarithmic radius once nearly tight rows and rows of large load are frozen (\Cref{sec:curvature}). Each direction's normal is added as a constraint, so directions within a batch are nearly orthogonal and, with the fair signs, the walk stays in its window. A batch ends after about $s/\log n$ constraints, so an \emph{epoch}, which ends when the number of active coordinates halves, contains only polylogarithmically many anchors.
 Within a batch, directions come from a square \emph{Gaussian completion} matrix $\mathbf Q$ whose rows are the constraints so far, hidden Gaussian rows, and one reference row. The vector $\mathbf Q^{-1}e_k$ is annihilated by every row except the reference row, and since the common kernel of $d-1$ Gaussian rows in $d$ dimensions is a uniformly random line, it yields a nearly isotropic feasible direction (\Cref{thm:sampler-conditioning}). Adding a constraint overwrites one hidden row, in a predetermined order, which gives the setting of van den Brand, Nanongkai, and Saranurak \cite{BNS2019}, and amortized $\tO(s^{\omega-1})$ work per operation (\Cref{thm:dynamic-maintenance}).\looseness=-1

 Finally, we note that retiring rows of tiny mass leaves $\tO(s)$ rows per anchor, so each epoch costs $\tO(s^\omega)$, and the geometrically shrinking epochs sum to $\tO(n^\omega)$. \Cref{sec:precision} implements the algorithm with polylogarithmic-bit numbers, using gap-free approximate diagonalization \cite{Shah2024} at anchors and dithered rounding that couples the sampler to an exact reference process. From here, independently restarting the process yields the expected-time bound.

 \subsection{Related Work}
Bansal's semidefinite-programming algorithm \cite{Bansal2010} and the
Lovett-Meka walk \cite{LovettMeka2015} made classical discrepancy bounds
constructive. Levy, Ramadas, and Rothvoss \cite{LRR2017} obtained deterministic
algorithms through multiplicative weights. Bansal and Garg \cite{BG2017}
controlled the full rounding trajectory to avoid the accumulation of
partial-coloring errors, while Bansal \cite{Bansal2024} combined iterated
rounding with concentration. The stochastic framework of Bansal, Laddha, and
Vempala \cite{BLV2022} and the regularization framework of Pesenti and Vladu
\cite{PV2023} give broader ways to design walks and discrepancy potentials.
Banaszczyk's Gaussian-measure theorem \cite{Banaszczyk1998} gives an
$O(\sqrt{\log(2n)})$ Koml\'os bound. Bansal, Dadush, and Garg
\cite{BDG2019} made this bound algorithmic. Dadush, Garg, Lovett, and Nikolov
\cite{DGLN2019} related vector balancing to subgaussian signed sums, and the
Gram-Schmidt walk \cite{BDGL2019} efficiently constructs such distributions.
More recently, Bansal and Jiang \cite{BJ2026} improved the Koml\'os bound to
$\widetilde O((\log n)^{1/4})$ through affine spectral independence. Guillen
and Kobzar \cite{GK2026} proved a constant bound when signs may be arbitrary
complex numbers of unit modulus. Their concave energy barriers are among
the precursors of the spectral signing framework used here.
On the other hand, fast implementations pose a separate problem. The self-balancing walk of
Alweiss, Liu, and Sawhney \cite{ALS2021} gives logarithmic Koml\'os bounds in
linear time, and Larsen \cite{Larsen2023} gives faster algorithms with
hereditary guarantees. Our analysis instead combines an absolute
discrepancy bound with nearly matrix-multiplication time. Its computational
tools come from dynamic inverse maintenance \cite{BNS2019} and the
matrix-formula framework \cite{Brand2021}. Stable fast linear algebra
\cite{DDH2007} and the bit-complexity analysis of dynamic algebraic formulas
\cite{ABGZ2024} address numerical implementation. These tools become useful
here after the geometric analysis limits how often the walk rebuilds its
curvature subspace.\looseness=-1
 
\subsection{Notation}
\label{sec:notation}
 
We use Euclidean vector norms and operator matrix norms unless a subscript indicates otherwise. Matrices are bold. For quadratic forms, $\mathbf{D}[h]$ means $h^\top \mathbf{D}h$. Matrix inequalities use the Loewner order unless we explicitly specify entrywise comparison, and $\ones$ denotes the all-ones vector. We write $L=\lceil\log_2(2n)\rceil$. Several objects come from Li \cite{Li2026}, sometimes under different names; for example, Li's $R_0$, $b_0$, and $A_0$ are our $R_{\max}$, $b_{\mathrm{dep}}$, and $\nu(0)$, and Li's $L$ is a top eigenvalue. Symbols introduced inside a proof in the appendices are local to that proof.

  The potential is $F_b$ (which we shorten to $F$ in Appendices~\ref{sec:appendix-geometry}, \ref{app:walk}, and~\ref{app:precision} where $b$ is fixed). The calligraphic $\mathcal F$ is the set of frozen rows when it is an index set (\Cref{sec:curvature}) and the exposed history in conditional expectations, and the italic $F$ is the feasible subspace in \Cref{sec:sampler} and Appendix~\ref{app:sampler}. Besides the spectral level $b$ (or $b_j$ in epoch $j$) and the coefficient vectors $b_i$, $b$ counts constraint slots in \Cref{sec:sampler}. The index $j$ ranges over coordinates, except in $b_j$, $N_j$, $T_j$, where it indexes epochs. $P$ and $\mathcal P$ are caps, and $\mathbf P$ is a curvature matrix. We let $\bm{\mathcal A}=b\mathbf I-\mathbf M$, where $\bm{\mathcal A}(\mathbf Q)$ is the augmented completion matrix. We use $t$ to index steps, $t_t$ to be a step length, and $t_i=z_i^\top v$ to be a load. Then, $K$ bounds $\|\mathbf Q\|_F,\|\mathbf Q^{-1}\|_F$ in \Cref{thm:dynamic-maintenance} and Appendix~\ref{app:dynamic} and abbreviates $K_{\mathrm{rat}}$ in Appendix~\ref{app:rows}. Inside proofs in Appendices~\ref{app:sampler} and~\ref{app:dynamic}, $L$ may name a local quantity; elsewhere $L=\lceil\log_2(2n)\rceil$. Finally, $G_i$ is the \emph{energy} of row $i$ and $\norm{x}^2$ the \emph{coordinate energy}; unqualified ``energy'' in \Cref{sec:walk} and Appendix~\ref{app:precision} means the latter. (Li uses that name for $\sum_j(1-x_j^2)=n-\norm{x}^2$.)

\section{The batched signing algorithm}
\label{sec:algorithm}
  We begin by recalling the framework we use from Li (\Cref{sec:foundations}), which gives the new batch certificate (\Cref{sec:curvature}), and assembles the algorithm.
 
\subsection{The framework from Li}
 
\label{sec:foundations}
 
The walk maintains a fractional coloring $x\in[-1,1]^n$, starting at $x=0$.
A coordinate is \emph{fixed} once it reaches $\pm1$, and $V$ denotes the set of
$s$ active coordinates. This subsection follows Li \cite[Section~3]{Li2026},
except for the retirement of rows with very small mass and a rational weight
function. Rows with $\|a_i\|_1\le1$, whose discrepancy never exceeds one, are
discarded first, and fewer than $n^2$ rows remain (Appendix~\ref{app:rows}), at a
cost of $O(mn)$ arithmetic operations.

Row $i$ is represented by an affine function $\varphi_i(y)=c_i+b_i^\top y$ with
\emph{tracked value} $d_i=\varphi_i(x)$. Initially, we have $b_i=a_i$ and $c_i=0$.
Moving the term $b_{ij}x_j$ into the offset ($c_i\gets c_i+b_{ij}x_j$,
$b_{ij}\gets0$) leaves $d_i$ unchanged and removes coordinate $j$ from the
row's future motion. This costs nothing if $j$ is fixed. If $j$ is active, the
true and tracked rows later differ by $b_{ij}(y_j-x_j)$, at most $2|b_{ij}|$,
and Appendix~\ref{app:rows} adds up these defects at the final signing.
 Now, define the squared coefficients $p_{ij}=b_{ij}^2$, the \emph{mass}
$q_i=\sum_{j\in V}p_{ij}$, and the \emph{energy}
$G_i(x)=\sum_{j\in V}p_{ij}(1-x_j^2)$, which measures how much the row can
still move. Because the columns of $A$ have norm at most one, we have
$\sum_iq_i\le s$, $\sum_ip_{ij}\le1$, and $|b_{ij}|\le1$, while
$0\le G_i\le q_i$. Set $R_{\max}=291/200$. A nonretired row is \emph{large} if
$q_i>R_{\max}^2$, \emph{medium} if $0<q_i\le R_{\max}^2$, and
\emph{completed} if $q_i=0$. Here, large rows are held in place: their tracked values
stay zero until they become medium. Li achieves this by repairing offsets. In our algorithm, all directions satisfy $b_i^\top h=0$ for every large row. Since $\sum_iq_i\le s$, there are fewer
than $s/R_{\max}^2$ large rows.

At each anchor, fixed-column entries are transferred to the offsets, and
medium-row entries with $p_{ij}>q_i/64$ are deleted, that is, transferred even
though their coordinate is active, in a fixed order. Large rows undergo no
deletions, mass never increases, and every surviving medium row satisfies the
\emph{spread condition} $\max_jp_{ij}\le q_i/64$. Each medium row also has a
nonincreasing rational \emph{scale} $r_i$ with $q_i\le r_i^2<\gamma^2q_i$ and
$r_i\le R_{\max}$, where $\gamma=65/64$; Appendix~\ref{app:rows} gives a choice
using $O(\log(2s))$ rational comparisons.
Li avoids evaluating every row at every step by means of a weak charge on the
tracked values and a motion clock. Instead of this approach, we stop tracking rows whose mass is
tiny. Choose $q_*=c_*\log^{-4}(2n)$ for a sufficiently small absolute
$0<c_*<R_{\max}^2$. After transfers and deletions, rows with $0<q_i\le q_*$
are \emph{retired} from spectral and derivative calculations: record
$b_i^{\rm ret}$ and $x^{\rm ret}$, transfer $b_i^\top x^{\rm ret}$ to $c_i$,
and set $b_i=0$. The later discrepancy of a retired row is a small random
error (\Cref{prop:walk-success}), and since $\sum_iq_i\le s$, at most $s/q_*$
nonempty rows remain.\looseness=-1
 
\Cref{alg:row-update} collects these updates. It runs at every anchor, after
the frozen weights of the preceding batch are released (\Cref{sec:curvature}),
and the spectral matrices are built afterward.
 
\begin{algorithm}[htbp]
\caption{Updating affine rows at an anchor}
\label{alg:row-update}
\small
\begin{algorithmic}[1]
\Require Coloring $x$, active set $V$, remaining affine rows, {cached masses and coefficient trees,} existing scales, and cutoff $q_*$.
\For{each remaining row $i$}
  \State Transfer fixed-column entries to the offset, subtract their original squared coefficients from $q_i$, update \\
  \qquad their tree entries, and classify row $i$.
  \While{$0<q_i\le R_{\max}^2$ and some $b_{ij}^2>q_i/64$}
    \State Transfer the first eligible entry to the offset, subtract its original squared coefficient from $q_i$, and \\ \qquad\quad update its tree entry.
  \EndWhile
  \State Complete row $i$ if $q_i=0$, or update $r_i$ by the scale rule if $0<q_i\le R_{\max}^2$.
  \State If $0<q_i\le q_*$, record $(b_i,x)$, set $c_i\gets c_i+b_i^\top x$ and $b_i\gets0$, and mark the row retired.
\EndFor
\end{algorithmic}
\end{algorithm}
  
Each medium row carries two signed barriers, one for each direction of its
discrepancy. Set the discrepancy threshold $H_0=263/5$, and normalize each
medium row's squared coefficients as $z_i=p_i/r_i^2$. For sign
$\sigma\in\{-1,1\}$, the barrier is
\begin{equation}
u_{i\sigma}(x)
 =\frac{\sigma d_i-H_0}{r_i}
  +\frac{4G_i(x)}{r_i^2}+12.
\label{eq:barrier-definition}
\end{equation}
The energy term reserves room for fluctuations the row has not yet seen, and
it makes the barrier concave in $x$. Nonpositive barriers give
$\sigma d_i\le H_0-12r_i-4G_i/r_i<H_0$ for both signs, hence $|d_i|<H_0$. At
fixed representation and scale,
$g_{i\sigma}=\nabla u_{i\sigma}=\sigma b_i/r_i-8z_i\circ x$ and
$\nabla^2u_{i\sigma}=-8\operatorname{diag}(z_i)$, so\looseness=-1
\begin{equation}
u_{i\sigma}(x+y)-u_{i\sigma}(x)
 =g_{i\sigma}(x)^\top y
       -4\sum_jz_{ij}y_j^2.
\label{eq:barrier-quadratic}
\end{equation}
In particular, a step orthogonal to $g_{i\sigma}$ can only lower the barrier,
which is how protected rows are kept safe. Spread gives $z_{ij}\le1/64$ and
$\sum_jz_{ij}\le1$. Hence
$|g_{i\sigma,j}|\le\sqrt{z_{ij}}+8z_{ij}\le2\sqrt{z_{ij}}$, so
$g_{i\sigma,j}^2\le4z_{ij}$ and $\norm{g_{i\sigma}}_2\le2$. Also,
$u_{i,+}+u_{i,-}\le32-2H_0/R_{\max}<-2$, so at most one signed barrier per row
lies in $[-1,0]$. With $\Lambda=2/q_*$, retirement bounds the load across both
signs by
$\sum_{i,\sigma}z_{ij}=2\sum_{i\text{ medium}}p_{ij}/r_i^2\le(2/q_*)\sum_ip_{ij}\le\Lambda$.

Li combines the barriers through a weighted Gram matrix whose weights grow as the barriers approach zero, and its weights are exponential. As our second departure, we use the rational weight with fixed degree $K_{\mathrm{rat}}=10000$,
defined for all $u\leq 0$ by
\begin{equation}
\nu(u)=119\left(1-\frac{u}{2K_{\mathrm{rat}}}\right)^{-K_{\mathrm{rat}}}.
\label{eq:rational-weight}
\end{equation}
It approaches $119e^{u/2}$ as $K_{\mathrm{rat}}$ grows, and it uses only field
operations and integer powers. Medium rows
contribute weighted rank-one terms, and large rows contribute diagonal
deposits for their later transition to medium. Set $b_{\mathrm{dep}}=3/625$ and
$\mathbf{B}=b_{\mathrm{dep}}\sum_{i\text{ large}}\operatorname{diag}(p_i)$.
Define the spectral matrix $\mathbf{M}=\mathbf{W}+\mathbf{B}$ with
\begin{equation}
\mathbf{W}=\sum_{i\text{ medium},\,\sigma}\nu(u_{i\sigma})z_iz_i^\top.
\label{eq:spectral-matrix}
\end{equation}
These matrices are PSD and entrywise nonnegative, with
$\mathbf{B}\preceq b_{\mathrm{dep}}\mathbf{I}$. For row updates, express each
signed contribution as $\nu(u_{i\sigma})p_ip_i^\top/r_i^4$.
Appendix~\ref{app:rows} shows that initialization gives
$\mathbf M\preceq b_{\rm dep}\mathbf I$ with all medium barriers below $-20$.
Later row updates preserve tracked sums. Deletion and scale reduction decrease
contributions entrywise, deposit replacement and retirement decrease them in
Loewner order, and removing fixed coordinates takes a principal submatrix.
An entrywise decrease need not be a Loewner decrease, so both comparisons are
used below.
 
As in Li, the final constant is the barrier threshold $H_0$ plus the total
deletion defect $2R_{\max}(8+\sqrt{63})$, which stay below $31673/320<99$
(\Cref{prop:final-defect-bound}). Li controls $\mathbf M$ through its top
eigenvalue; we use the resolvent potential $F_b$ of \eqref{eq:intro-potential}
instead, which the monotonicity facts above make nonincreasing under row
updates (\Cref{lem:curv-ground-state}).
 
\subsection{A subspace for an entire batch}
 
\label{sec:curvature}
 
This subsection contains the main new ingredient. At an anchor, we find many directions along which the resolvent potential is concave, and then show that the same directions remain valid throughout a batch of steps. Fix the row representations and scales at an anchor $x_0$. For a signed medium row $i=(\ell,\sigma)$, write $u_i=u_{\ell\sigma}$, $g_i=\nabla u_i$, and $z_i=z_\ell$. An anchor is \emph{admissible} when $1\le b<3/2$, $\mathbf{M}_0\prec b\mathbf{I}$, and every $u_i(x_0)\le0$. Write
\[
 \bm{\mathcal A}=b\mathbf{I}-\mathbf{M},\qquad
 \mathbf{R}=\bm{\mathcal A}^{-1},\qquad
 v=\mathbf{R}\ones,\qquad
 F_b(\mathbf{M})=\ones^\top v-s/b.
\] We use the potential in \eqref{eq:intro-potential}. The resolvent is entrywise nonnegative, and $F_b\ge0$ at an admissible anchor.
 
Suppose that $\sum_i z_{ij}\le\Lambda$ and $\ones^\top v_0\le\mathcal P s$. Set $\rho=1/64$, $\Theta=4096\Lambda\mathcal P$, and $\kappa=\max\{1,119\Lambda\Theta\}$. The load of row $i$ is $t_i=z_i^\top v$. Protect every row with $u_i(x_0)\ge-1$, and select heavy rows, each with $t_i(x_0)\ge\Theta/2$, so that every remaining row has load at most $\Theta$. Let $\mathcal F$ be the protected or heavy signed rows and $\mathcal L$ the others. Form $\mathbf H$ from the large-row vectors $b_\ell^\top$ and the gradients $g_i(x_0)^\top$ for $i\in\mathcal F$. Changes in $\ker\mathbf H$ preserve large-row sums and keep the fixed rows' barriers at most their anchor values.
Fix the weights in $\mathcal F$ at their anchor values.
 
Each kind of frozen row is frozen for its own reason. Protected rows are frozen so that every light row starts at least one unit below its threshold, and therefore stays negative throughout the window defined below. Heavy rows are frozen because the stability argument compares the resolvent vector at $x$ with the one at $x_0$, and a light row's influence on that vector is proportional to its load. Bounding every light load by $\Theta$ is what keeps the constant $\kappa$ polylogarithmic. Freezing is cheap: the counting in Appendix~\ref{sec:appendix-geometry} shows that protected and heavy rows remove only a small fraction of the dimensions. With the frozen weights fixed, the batch matrix is\looseness=-1
\begin{equation}
 \mathbf{M}(x)=\mathbf{M}_0+
   \sum_{i\in\mathcal L}[\nu(u_i(x))-\nu(u_i(x_0))]z_iz_i^\top.
 \label{eq:batch-fixed-matrix}
\end{equation}
Its curvature splits into a positive term $\mathbf P$ and a diagonal negative term $\mathbf D$:
\begin{equation}
\begin{aligned}
 \mathbf U&=\sum_{i\in\mathcal L}\nu''(u_i)t_i^2g_ig_i^\top,
 &\mathbf D&=8\sum_{i\in\mathcal L}\nu'(u_i)t_i^2\diag(z_i),\\
 \mathbf C&=\sum_{i\in\mathcal L}\nu'(u_i)t_i z_ig_i^\top,
 &\mathbf P&=\mathbf U+2\mathbf C^\top\mathbf R\mathbf C.
\end{aligned}
 \label{eq:batch-curvature}
\end{equation}
Where $\mathbf M(x)\prec b\mathbf I$, the Hessian is $\nabla^2F_b(\mathbf M(x))=\mathbf P-\mathbf D$; Appendix~\ref{sec:appendix-geometry} gives the calculation. The subscript $0$ denotes an evaluation at $x_0$. To bound coordinate changes in the light barriers, we use the convex window
\begin{equation}
 \mathcal N_a(x_0)=\{x_0+y\in[-1,1]^s:
   \norm y_\infty\le a,\ |g_i(x_0)^\top y|\le a, \forall i\in\mathcal L\}.
 \label{eq:batch-window}
\end{equation}
The certificate below asks for $\mathbf P_0[h]\le(1-\rho)\mathbf D_0[h]$ on a subspace, rather than only $\mathbf P_0[h]\le\mathbf D_0[h]$, which would already give concavity at $x_0$. The margin $\rho$ is the slack that absorbs the drift of $\mathbf P$ and $\mathbf D$ as the walk moves through the window, and it is what lets one computation at $x_0$ serve a whole batch.
 
\begin{proposition}[A valid subspace throughout a batch]
\label{prop:batch-geometry}
At an admissible anchor with the row invariants above and the stated load and resolvent bounds, there exists a subspace $E\subseteq\ker\mathbf H$ with $\dim E\ge\lceil s/32\rceil$ and $\mathbf P_0[h]\le(1-\rho)\mathbf D_0[h]$ for all $h\in E$. Here, one may choose any $\lceil s/32\rceil$-dimensional nonpositive subspace of $\mathbf P_0-(1-\rho)\mathbf D_0$ restricted to $\ker\mathbf H$.
For a universal $c>0$, choose $0<a\le c\rho/(1+\kappa)^2$. Then, throughout $(x_0+E)\cap\mathcal N_a(x_0)$, the batch matrix satisfies $\mathbf M(x)\prec b\mathbf I$, fixed weights majorize their true weights, light barriers are negative, and for all $h\in E$, we have that
\[ \mathbf P_x[h]\le(1-\rho/2)\mathbf D_x[h].\]
So, $F_b(\mathbf M(x))$ is concave on this set. If $x,x+th,x-th$ belong to it, where $h\in E$ and $t\ge0$, then an independent fair sign $\sigma$ satisfies $\E_\sigma F_b(\mathbf M(x+\sigma th))\le F_b(\mathbf M(x))$, and the step size may depend on $x$ and $h$ before choosing $\sigma$.
\end{proposition}
\noindent We include the proof in Appendix~\ref{sec:appendix-geometry}.

\subsection{Batched Spectral Signing}
 
The algorithm chooses a subspace at each anchor. During the batch that
follows, only the coloring and the completion sampler change; row
representations, scales, and the anchor spectral matrices stay fixed. Set $L=\lceil\log_2(2n)\rceil$ and
$\rho=1/64$. Choose the potential caps $P=C_PL^3$ and
$\mathcal P=2(P+1)$, and the retirement cutoff $q_*=c_*L^{-4}$.
These give the load parameters $\Lambda=2/q_*$,
$\Theta=4096\mathcal P\Lambda$, and
$\kappa=\max\{1,119\Lambda\Theta\}$. The radius and step length satisfy $a=c_a\rho(1+\kappa)^{-2}$ and $\frac{a}{2^{21}L}<\tau\le\frac{a}{2^{20}L}$. Choose a dyadic $a$ within a factor two of this value and a power-of-two
$\tau$ in this interval. Choose $C_P,C_B$ sufficiently large and $c_a,c_*$
sufficiently small for the bounds below. Replacing the logarithm in the
cutoff from Section~\ref{sec:foundations} by $L$ makes it rational and changes only $c_*$.
Epoch $j$ starts with $N_j$ active coordinates and uses the spectral
barrier $b_j=1+\frac14\sum_{\ell=1}^j\ell^{-2} \in [1,3/2)$, ending when at most $N_j/2$ coordinates are active. Its exposure cap
is $T_j=C_TN_j\tau^{-2}L$, and its anchor potential cap is $N_jP$. A trial also restarts on the finite-precision failures described in Section~\ref{sec:precision}.\looseness=-1
 
\begin{algorithm}[t]
\caption{Batched spectral signing with exact verification}
\label{alg:main}
\small
\begin{algorithmic}[1]
\Require Rational $\mathbf{A}\in\Q^{m\times n}$ with column norms at most $1$.
\Ensure A signing $\eps\in\{-1,1\}^n$ with $\|\mathbf{A}\eps\|_\infty<99$.
\State Round $\mathbf{A}$ toward zero with $\|\mathbf{A}-\bar {\mathbf{A}}\|_{\infty\to\infty}\le10^{-4}$. 
\State Discard rows of $\bar {\mathbf{A}}$ with $\ell_1$ norm at most one.
\State \textbf{if} {$n<64$ or no rows remain} then return any signing in $\{-1,1\}^n$.
\Loop \Comment{Each iteration is an independent trial.}
    \State Initialize $x=0$, $V=[n]$, and $(c_i,b_i)=(0,\bar a_i)$
    \State Initialize cached masses and coefficient trees, $j=0$, $N_0=n$, and exposure count $0$.
    \State Apply Algorithm~\ref{alg:row-update} and form $\mathbf{M}$, including the large-row deposits.
    \While{$|V|\ge64$}
        \If{$|V|\le N_j/2$}
            \State Set $j\gets j+1$ and $N_j\gets|V|$, use barrier $b_j$, and reset the exposure count.
        \EndIf
        \State \textbf{if} the exposure count reaches its cap \textbf{then restart trial}.
        \State Construct the anchor and sampler below, then run Algorithm~\ref{alg:batch}.
        \State Reset the fixed majorants, apply Alg~\ref{alg:row-update}, and remove fixed coordinates from the ambient matrices.
    \EndWhile
    \State Form $\eps$ by nearest-sign rounding, with a fixed tie rule at zero, and evaluate $\bar {\mathbf{A}}\eps$ exactly.
    \State \textbf{if} {$\|\bar {\mathbf{A}}\eps\|_\infty\le98.99$} \textbf{then return} $\eps$.
\EndLoop
\end{algorithmic}
\end{algorithm}
Every signing that \Cref{alg:main} returns therefore satisfies
\begin{equation}
 \norm{\mathbf{A}\eps}_\infty
 \le\norm{\bar {\mathbf{A}}\eps}_\infty+10^{-4}
 \le98.9901<99.
 \label{eq:algorithm-verification}
\end{equation}
Discarded rows obey the same bound with $1$ in place of $98.99$.
Section~\ref{sec:walk} proves that each capped trial succeeds with constant
probability.

 We now explain how to build an anchor. At an anchor $x_0$, compute the row statistics, $\mathbf{M}$, and its
resolvent. Protect each signed row with $u_i(x_0)\ge-1$, and mark a row heavy
when its load $t_i=z_i^\top v$ is at least $\Theta/2$; any rule under which
heavy rows have load at least $\Theta/2$ and the others at most $\Theta$
would do. Form $\mathbf{H}$ from
the large-row vectors and the gradients of protected or heavy rows, and fix
these rows' weights at their anchor values. Compute a subspace $E$ of
dimension $k=\lceil s/32\rceil$ as in \Cref{prop:batch-geometry}, that is, a
nonpositive subspace of $\mathbf{P}_0-(1-\rho)\mathbf{D}_0$ restricted to
$\ker\mathbf H$, and obtain its sampling basis $\mathbf{Z}$ by the small
Gaussian perturbation of \Cref{sec:sampler}. Give the batch
$\max\{1,\lfloor k/(C_BL)\rfloor\}$ constraint slots, with insertion count
initially zero. For $k\le C_0L$, choose an orthonormal basis $\mathbf U_E$ of
the batch span, sample $h=\mathbf U_Eg/\|g\|_2$ with
$g\sim N(0,\mathbf I_k)$, and end the batch after one proposal. Otherwise
form the hidden rows $\mathbf{G}\mathbf{Z}$ and initialize the augmented
inverse from \eqref{eq:sampler-augmentation}. In \Cref{sec:precision}, we explain how these steps can be carried out in finite precision.

\begin{algorithm}[t]
\caption{One batch of coloring moves}
\label{alg:batch}
\small
\begin{algorithmic}[1]
\Require The current coloring, epoch counters, and the anchor and sampler.
\Ensure An updated coloring and epoch counters, or trial failure.
\Loop
    \State {Obtain unit $h$ by isotropic sampling in the batch span or \eqref{eq:sampler-output-query}, together with $c=\mathbf Z^\top h$.}
    \State Quantize $(h,c)$ independently to $(\bar h,\bar c)$; set $\widehat h=3\bar h/4$ and mask fixed coordinates.
    \State Set $t\gets\min\left\{\tau,\min_{j\text{ active},\,\widehat h_j\ne0}
             (1-|x_j|)/|\widehat h_j|\right\}$, draw a fresh fair sign $\sigma$, and move to $x+\sigma t\widehat h$.
    \State Round the new point to the nearest dyadic point, preserving attained values $\pm1$.
    \State Update $x$, record newly fixed coordinates $J$, set $V\gets V\setminus J$, and increment the exposure count.
    \State Abort if the exposure cap is reached before the epoch ends.
    \State \Return if the epoch ends, $|V|<64$, sampling is explicit, or fewer than $1+|J|$ insertion slots remain.
    \For{the self-normal $\bar c^\top$, then each $\mathbf{Z}_{j,:}$ with $j\in J$ in increasing order}
        \State Refresh the reference row, replace the next hidden row, and increment the insertion count.
    \EndFor
    \State \Return if no insertion slots remain.
\EndLoop
\end{algorithmic}
\end{algorithm}
 
Fixed coordinates are removed from the ambient matrices at the next anchor.

\section{Analysis of the walk}
\label{sec:walk}
 
The analysis must show three things: coordinates reach $\pm1$ within few
steps, each batch stays in the region where its subspace is valid, and the
potential stays bounded. This section treats the walk in exact arithmetic,
assuming the anchor and sampler guarantees; numerical errors are handled in
\Cref{sec:precision}. We analyze a capped trial, stopping at the first
invalid proposal and setting all later increments to zero.
 
Throughout, the endpoint sign of a step is drawn independently, after its
direction and clipping length have been exposed. This order is essential.
Midpoint concavity and every martingale bound below require the sign to be
independent of the step it signs. The clipping length treats both endpoints
symmetrically, so it can be computed before the sign is known.
 
\Cref{prop:batch-geometry} gives a subspace whose concavity persists
throughout the batch window. We now bound the number of steps and show that
the walk stays in this window.
 
\subsection{Progress and batch counts}
\label{subsec:walk-progress}
 
A symmetric move increases the expected squared distance from the origin, and
when clipping shortens a move, at least one endpoint fixes a coordinate.
Together, these two facts bound the number of queries.
 
\begin{lemma}[Progress with clipped symmetric steps]
\label{lem:walk-query-count}
Start with at most $N$ active coordinates. Use directions in $E$ with
norm in $[1/3,1]$ that vanish on fixed coordinates, and independent fair
endpoint signs. Stop at epoch completion or failure. Then, we have that
\begin{align}
 \mathbb E\sum_t t_t^2\|h_t\|^2&\le N,\label{eq:walk-energy}\\
 \mathbb E\#\{t:t_t=\tau\}&\le9N\tau^{-2},\label{eq:walk-full}\\
 \mathbb E\#\{t:t_t<\tau\}&\le2N.\label{eq:walk-short}
\end{align}
For an absolute $C_T$, the probability that the epoch needs more than $C_TN\tau^{-2}L$ queries is at most $n^{-10}$.\looseness=-1
\end{lemma}
\begin{proof}
Count all initial coordinates, including those fixed later. Conditional
on the history and exposed direction, the expected energy is
$\frac12\|x+th\|^2+\frac12\|x-th\|^2=\|x\|^2+t^2\|h\|^2$, which lies in $[0,N]$. We sum the stopped identity over a bounded number of
steps and apply monotone convergence for \eqref{eq:walk-energy}.
Full steps contribute at least $\tau^2/9$, giving \eqref{eq:walk-full}.
A short step has a clipping minimizer that at least one endpoint fixes.
Let $J_t$ indicate a short step and $Z_t$ count newly fixed coordinates.
Then $\mathbb E[Z_t|\mathcal F_t,h_t]\ge J_t/2$ and
$\sum_tZ_t\le N$, so
$\frac12\mathbb E\sum_tJ_t\le\mathbb E\sum_tZ_t\le N$, proving \eqref{eq:walk-short} for any clipped step length.
From any later history, this also bounds expected remaining
queries by $11N\tau^{-2}$. For $D=\lceil22N\tau^{-2}\rceil$ and query
count $T$, conditional Markov bounds give
$\Pr[T>(r+1)D| T>rD]\le1/2$ and
$\Pr[T>rD]\le2^{-r}$.
Take $r=\lceil20L\rceil$ and increase $C_T$ for the progress lemma.\looseness=-1
\end{proof}
 
Starting dimensions decrease geometrically, so the capped trial uses
$O(n\tau^{-2}L)=n\log^{O(1)}n$ direction queries.
In an epoch of initial size $N$, $T$ exposures use at most $T+N$ slots:
one self-normal per exposure and at most $N$ coordinate normals.
Ordinary batches consume $\Omega(N/L)$ slots. If a coordinate group
would exceed the remaining slots, charge that batch's termination to the
group. Therefore\looseness=-1
\begin{align}
 \#\{\text{batches in the epoch}\}
 &\le C\left(1+\frac{L(T+N)}{N}\right)\le C\left(1+L+C_T\tau^{-2}L^2\right)
 =\log^{O(1)}n.
 \label{eq:walk-batches}
\end{align}
For explicit one-direction batches, $k=O(L)$ and $s_t\ge N/2$ imply
$N=O(L)$, so their count is at most
$T\le C_TN\tau^{-2}L=O(\tau^{-2}L^2)$.
Thus each capped epoch has $O(\tau^{-2}L^2)$ batches. There are at most
$1+\lceil\log_2n\rceil$ epochs, giving
$B_{\rm tot}\le C\tau^{-2}L^3$ batches per capped trial.
 
\subsection{Concentration and final discrepancy}
 
Self-constraints make successive directions nearly orthogonal within a
batch. If $\mathbf H_b=[h_1\ \cdots\ h_b]$, their Gram bound
$\|\mathbf H_b^{\mathsf T}\mathbf H_b\|_{\mathrm{op}}\le2$ controls every
coordinate and barrier gradient at once (see \eqref{eq:walk-gram}). Adding each
direction as a constraint is what yields this bound, and it is a single
operator-norm bound covering all functionals at once, so the same estimate
handles coordinates, barrier gradients, and retired rows, which the algorithm
never examines.
Concavity then makes the spectral potential a supermartingale, and the small
norm of a retired row limits its later fluctuation even across many batches.
The inner window is
$(x_0+E)\cap\mathcal N_{a/4}(x_0)$; the step length keeps both proposed
endpoints in the outer window.
 
\begin{proposition}[Trial guarantees]
\label{prop:walk-success} The following hold
jointly for the stopped capped walk with probability at least
$1-1/100-n^{-8}-n^{-6}$: every batch stays in its inner trust window.
In epochs of initial size $N$, every query state satisfies
$F_{b_j}(\mathbf M_t)\le NP$ and
$\mathbf1^{\mathsf T}(b_j\mathbf I-\mathbf M_t)^{-1}\mathbf1
\le\mathcal P s_t$, where $s_t$ is the active dimension before the query.
The first bound also holds at the terminal endpoint before row updates.
Every row's error after retirement is at most $1/100$, including
nearest-sign rounding of the last fewer than $64$ coordinates.\looseness=-1
\end{proposition}
 
Appendix~\ref{app:walk} proves the proposition. The window bound uses the
batch Gram inequality and independent endpoint signs. In a valid window the
potential is concave, and row updates and epoch barrier increases cannot
increase it, so Doob's inequality controls every epoch. For a retired row,
the squared projections sum to $O(q_*L^3)$ over the capped trial, so
sufficiently small $c_*$ proves the error bound.
Proposition~\ref{prop:final-defect-bound} combines these guarantees with the
barrier invariant to bound the final discrepancy. Section~\ref{sec:precision}
adds numerical error and assembles the constant probability of passing
exact verification.\looseness=-1
 
\section{Adaptive completion sampling}
\label{sec:sampler}
 
The walk needs directions that vanish on newly fixed coordinates and are
almost orthogonal to earlier directions, so that cumulative motion remains
controlled. Recomputing a feasible basis after every step would be costly.
Instead, we impose each new constraint by replacing one row of a Gaussian
completion matrix and obtain the next direction from an inverse column.
Within each batch, fix $\mathbf{Z}\in\R^{s\times k}$ satisfying $k\ge c_0s$ and
\begin{equation}\label{eq:sampler-basis-conditioning}
 \frac12\le \sigma_{\min}(\mathbf{Z})\le\sigma_{\max}(\mathbf{Z})\le2.
\end{equation}
We sample in $E=\operatorname{range}(\mathbf{Z})$ through coefficient
coordinates. Each exposed direction is added back as a constraint, which we
call a \emph{self-constraint}, alongside the constraints for newly fixed
coordinates. Form the first $k-1$ rows of $\mathbf{Q}\in\R^{k\times k}$ as
$g_i^\top\mathbf{Z}$ for independent $g_i\sim N(0,\mathbf{I}_s)$ in one
matrix product. Its last, \emph{reference} row is an independent
$r\sim N(0,\mathbf{I}_k)$, refreshed with $O(k)$ scalar samples.
Replace Gaussian rows in fixed order, calling the survivors \emph{hidden}.
For invertible $\mathbf{Q}$, $y=\mathbf{Q}^{-1}e_k$ annihilates its first
$k-1$ rows. The reference row orients and scales $z=\mathbf{Z}y$ on their
null line. With $w=\mathbf{Z}^\top z$, we normalize by
\begin{equation}\label{eq:sampler-output-query}
 h=\frac{z}{\lVert z\rVert_2},
 \qquad c=\frac{w}{\lVert z\rVert_2}=\mathbf{Z}^\top h.
\end{equation}
 
In exact arithmetic one may take $\bar h=h$ and $\bar c=c$. In finite
precision, the dyadic rounder of \Cref{sec:precision} uses independent dithers
to produce $\bar h=\mathcal Q(h)$ and $\bar c=\mathcal Q(c)$ with
$\lVert\bar h-h\rVert_2\le\epsilon_s$ and
$\lVert\bar c-\mathbf{Z}^\top h\rVert_2\le\epsilon_s$.
The self-constraint $\bar c^\top y'=0$ has physical normal
$n=\mathbf{Z}(\mathbf{Z}^\top\mathbf{Z})^{-1}\bar c\in E$, satisfying
$\lVert n-h\rVert_2\le2\epsilon_s$.
Following Algorithm~\ref{alg:batch}'s insertion order and capacity check,
each real insertion independently refreshes the reference row and replaces
the next hidden row.
Let $F\subseteq E$ be the common kernel of the real constraint normals,
and write $\mathcal F$ for the exposed history. This history contains directions, quantized outputs, walk decisions,
and real constraints. It excludes hidden rows, raw reference rows,
unnormalized outputs $(y,z,w)$, and other inverse entries used internally.
Stop before the first dependent real insertion, as determined by this history.
Each earlier insertion removes one hidden row and one feasible dimension,
leaving $\dim F-1$ hidden rows.
If $2\epsilon_s<1$, a self-normal cannot cause the first dependence:
$h$ lies in the current feasible space $F$, previous real normals are
orthogonal to $h$, and $\langle n,h\rangle\ge1-2\epsilon_s>0$.
So only a coordinate insertion can cause the stop.\looseness=-1
 
Perturb the anchor basis $\mathbf{Z}_0$, whose singular values lie in
$[3/4,5/4]$, using $\eta=n^{-C_\eta}$:
\begin{equation}\label{eq:sampler-smoothed-basis}
 \mathbf{Z}=\mathbf{Z}_0+\eta \mathbf{G}+\mathbf{E}_{\mathrm{round}},
 \qquad \mathbf{G}_{ij}\overset{\mathrm{iid}}\sim N(0,1),
 \qquad \lVert \mathbf{E}_{\mathrm{round}}\rVert_{\mathrm{op}}
       \le\eta\sqrt{k}/32.
\end{equation}
Set $C_\eta$ within the anchor's allowed error, condition on $\mathbf{Z}$,
and draw an independent hidden completion.
For large constants $C_b,C_k$, set
$b=\lfloor k/(C_b\log(2n))\rfloor$. Choose $C_B,C_0$ so that
$C_BL\ge C_b\log(2n)$ and $C_0L\ge C_k\log(2n)$, with $C_0\ge C_B$.
Every nonexplicit batch then inserts at most $b$ constraints.
This cap comes from the sampler, not from the trust window: the window
argument uses only the Gram bound, which depends on the batch length only
through $b\epsilon_s$, a negligible quantity because $\epsilon_s$ is inverse
polynomial.
The sampler needs the feasible dimension to stay at least $k/2$, so that
directions remain nearly isotropic, and its union bounds over sets of
coordinates need $b$ to be smaller than $k$ by a logarithmic factor.\looseness=-1
 
\begin{theorem}[Completion-sampler guarantees]
\label{thm:sampler-conditioning}
Fix $A>0$ and a polynomial upper bound on sampler operations in a trial.
Choose the preceding constants sufficiently large and set
$\epsilon_s=n^{-C_s}$ with sufficiently large $C_s$.
For a basis satisfying \eqref{eq:sampler-basis-conditioning}, every query
of the stopped process, with $d=\dim F$, satisfies for predictable $v$
\begin{equation}\label{eq:sampler-subgaussian-square}
 \E\left[
 \exp\left(\frac{d\langle v,h\rangle^2}
                  {16\lVert v\rVert_2^2}\right)
 \middle|\mathcal F\right]\le2,
 \qquad
 \E[hh^\top\mid\mathcal F]\preceq\frac2d \mathbf{P}_F,
\end{equation}
where $\mathbf{P}_F$ is the orthogonal projection onto $F$ and the first
bound is trivial for $v=0$.
With probability at least $1-n^{-A}$ over the smoothed basis and sampler
randomness, the basis satisfies \eqref{eq:sampler-basis-conditioning}, no
real insertion is dependent,
and every query matrix and every intermediate matrix after a reference
refresh or real insertion satisfies
$\lVert\mathbf{Q}\rVert_F+\lVert\mathbf{Q}^{-1}\rVert_F\le n^{C_Q}$
for a fixed constant $C_Q$.
The conditional estimates use only the exposed history, without
conditioning on this high-probability event.
\end{theorem}
 
\noindent\emph{Proof idea.}
Exposing directions destroys Gaussian independence, but preserves a
rotationally invariant law after a well-conditioned change of coordinates.
Small self-normal errors change this conditioning only quadratically.
Smoothing separates coordinate constraints from self-normals, and the
remaining hidden rows have a quadratic singular-value lower tail.
These bounds give a polynomially bounded right inverse for the first
$k-1$ rows. The independent reference row resolves their null line by a
Gaussian small-ball bound. Appendix~\ref{app:sampler} gives the proof.
 
{In the explicit branch, use the isotropic direction specified
in Section~\ref{sec:algorithm}. Orthonormalization and one query cost
$\tO(s^\omega)$, which is polylogarithmic because
$k\le C_0L$ and $k\ge c_0s$.}
Return $z$ and $w$ without extra dense matrix-vector products by stacking
$\mathbf{B}=\binom{\mathbf{Z}}{\mathbf{Z}^\top\mathbf{Z}}$.
Let $\bm{\mathcal A}(\mathbf{Q})$ be the block lower-triangular matrix
with diagonal blocks $\mathbf{Q}$ and $\mathbf{I}_{s+k}$ and lower-left
block $-\mathbf{B}$.
Its inverse gives all three outputs in one column:
\begin{equation}
\label{eq:sampler-augmentation}
 \begin{aligned}
 \bm{\mathcal A}(\mathbf{Q})^{-1}
 =\begin{pmatrix}\mathbf{Q}^{-1}&0\\       \mathbf{B}\mathbf{Q}^{-1}&\mathbf{I}_{s+k}\end{pmatrix},\quad
 \bm{\mathcal A}(\mathbf{Q})^{-1}\binom{e_k}{0}
 =\begin{pmatrix}y & z &w\end{pmatrix}^\top.
 \end{aligned}
\end{equation}
Its dimension is $s+2k=O(s)$ and its norms remain polynomial because
$\lVert\mathbf{B}\rVert_{\mathrm{op}}=O(1)$.
Section~\ref{sec:dynamic} maintains this augmentation and includes its
construction and all sampler operations in the run time.\looseness=-1

\section{Dynamic matrix formula maintenance and running time}
\label{sec:dynamic}

Transpose \eqref{eq:sampler-augmentation} to turn row replacements into
column updates and inverse-column requests into full-row queries
\cite{Brand2021}, returning $y,z,w$ together. Because constraint slots are
filled in a predetermined order, the data structure has weak index lookahead
\cite{BNS2019}. The finite-precision theorem of \cite{ABGZ2024} assumes an
offline index sequence, so Appendix~\ref{app:dynamic} extends its additive
Woodbury error analysis to these candidate sets.
Given $\|\mathbf{Q}\|_F,\ \|\mathbf{Q}^{-1}\|_F\le K$
and $\|\mathbf{B}\|_F\le K_C$, the block identities give
\begin{align*}
 \|\bm{\mathcal A}(\mathbf{Q})\|_F^2
 =\|\mathbf{Q}\|_F^2+\|\mathbf{B}\|_F^2+s+k
 \le K^2+K_C^2+s+k,\quad
 \|\bm{\mathcal A}(\mathbf{Q})^{-1}\|_F^2
 \le K^2+K_C^2K^2+s+k.
\end{align*}
The near-isometry of $\mathbf{Z}$ gives $K_C=O(\sqrt s)$.
 
Each operation replaces a row or requests an inverse column, and it touches
only the reference slot $k$ or the next constraint slot $p_t$. Since $p_t$
increases by at most one per operation, the next
$q$ operations use
$\mathcal I(t,q)=\{k\}\cup(\{p_t,p_t+1,\ldots,p_t+q\}\cap[k])$,
which has at most $q+2$ indices. For $t\le t'$ and $t'+q'\le t+q$,
$p_t\le p_{t'}\le p_t+(t'-t)$ and $p_{t'}+q'\le p_t+q$, so
$\mathcal I(t',q')\subseteq\mathcal I(t,q)$.
Fixing several coordinates uses separate operations, meeting the condition of
\cite[Def 6.3 and Thm 6.8]{BNS2019}. \looseness=-1
 
\begin{theorem}[Augmented inverse maintenance]
\label{thm:dynamic-maintenance}
Fix $\mathbf{Z}\in\R^{s\times k}$ with $k\le s$, and assume update and query indices have nested candidate sets
of size $O(q)$, available at the start of each horizon $q$.
Suppose throughout the $T$ operations that
$\|\mathbf{Q}_t\|_F,\|\mathbf{Q}_t^{-1}\|_F\le K$ and
$\|\mathbf{B}\|_F\le K_C$.
For $\epsilon\in(0,1]$, a data structure returns the three unnormalized vectors in
\eqref{eq:sampler-augmentation} with entrywise error at most $\epsilon$ in\looseness=-1
\begin{equation}
 \tO\left(
   \bigl(s^\omega+Ts^{\omega-1}\bigr)
   \log\frac{s(1+K)(1+K_C)}{\epsilon}
 \right)
 \label{eq:dynamic-bit-cost}
\end{equation}
bit operations.
With exact arithmetic the data structure uses
$\tO(s^\omega+Ts^{\omega-1})$ arithmetic operations.
\end{theorem}
 
Appendix~\ref{app:dynamic} proves the theorem using a hierarchy of candidate
inverse rows and additive backward-error control.
\Cref{thm:sampler-conditioning} proves the norm hypothesis even
between a reference refresh and the following insertion, with adaptive row values.
 
\paragraph{Running time.}
Each exposure or insertion uses $O(1)$ dynamic operations and $O(s)$
scalar work in \Cref{alg:batch}.
By \Cref{sec:foundations}, there are at most $s/q_*$ live rows, since
$q_i>q_*$ and $\sum_iq_i\le s$. Including both signs gives
$\tO(s)$ signed rows since $q_*^{-1}=\operatorname{polylog}n$. {We maintain incidence lists and permanent contributions, and update
each mass $q_i$ by subtracting the squared coefficient whenever an entry
is removed. At trial initialization, we build a segment tree on each row's
list of nonzero entries in its fixed deletion order. Each node stores
the largest surviving squared coefficient in its interval.
Descending through intervals whose maximum exceeds $q_i/64$ finds the
first eligible entry in $O(\log(2n))$ work, deletion updates the tree in
the same time. We remove fixed coordinates at batch boundaries, since their
constraints force zero increments until then. Each input incidence leaves
its list once, through fixing, coefficient deletion, or retirement.
The trees take linear incidence work to build once per trial, and
processing all removals therefore costs
$\tO(\operatorname{nnz}(\bar{\mathbf A}))$ arithmetic operations.}
 
\Cref{prop:walk-success} and final verification handle retired rows.
We reconstruct live values, gradients, and weights at the next anchor in
$O(s^2/q_*)=\tO(s^2)$ work.
The light-row contribution to $\mathbf{M}$ and fixed-row majorants are
weighted Gram products. In \eqref{eq:batch-curvature}, $\mathbf{U}$ is another
weighted Gram product, $\mathbf{C}$ is a mixed product, and $\mathbf{D}$
is a diagonal row sum. Each product has $O(s/q_*)$ summands, so
dividing it into blocks of $s$ costs $O(q_*^{-1}s^\omega)=\tO(s^\omega)$.
Forming $\mathbf{P}$, inverting the spectral matrix, extracting the curvature
subspace, and initializing the augmented sampler have the same cost. Polynomial norm
bounds add only logarithmic factors in finite precision.
Within a batch, weight changes can amount to a rank-$s$ update of
$\mathbf M$; \Cref{prop:batch-geometry} lets us defer it to the next anchor,
while the completion matrix changes only by rank-one row replacements.
 
For an epoch starting with $N$ active coordinates, queries have
$N/2\le s\le N$ and polynomial norm bounds on the sampler's good event.
Equation~\eqref{eq:walk-batches} gives
$B_N=\operatorname{polylog}n$ anchors. The exposure and slot counts of
Section~\ref{subsec:walk-progress} give $T_N=\tO(N)$ dynamic operations,
including reference refreshes, constraint insertions, and direction
requests. Hence
\begin{align}
 \text{epoch work}
 &\le\tO\left(B_NN^\omega+T_NN^{\omega-1}
                         +N T_N+B_NN^2\right)=\tO(N^\omega).
 \label{eq:dynamic-epoch-cost}
\end{align}
 
The epoch sizes satisfy $N_j\le n/2^j$, so the total is
$\sum_{j\ge0}\tO(N_j^\omega)
\le\tO\left(n^\omega\sum_{j\ge0}2^{-j\omega}\right)
=\tO(n^\omega)$.
Incidence updates and verification cost $\tO(L_{\mathrm{in}}+mn)$, with
$mn=O(L_{\mathrm{in}})$ for dense encoding. For sparse encoding, they
cost its nonzero encoding length, plus $O(n)$ to write the output.
 
\section{Finite-bit Implementation and Stability}
\label{sec:precision}
This section implements the algorithm with numbers of polylogarithmically
many bits. Similarly to \cite{ABC2026finiteprecision}, we focus on the analysis of a randomized discrepancy algorithm under approximate arithmetic; however, our techniques here are far simpler as we are not concerned with a distributional guarantee such as subgaussianity \cite{BDGL2019, harshaw2023balancingcovariatesrandomizedexperiments}, and can appeal to a coupling with an ideal run of the algorithm.

Two tools do most of the work for this. First, the anchor subspace is
computed from a penalty matrix by an approximate eigendecomposition that needs
no spectral gap (\Cref{lem:penalty-bound}). Second, every stored output is
rounded with an independent random dither. This couples the finite algorithm
to an idealized reference process that uses exact Gaussians: with high
probability the two store exactly the same outputs, so the success
probabilities proved for the reference carry over.
 
Fix exponents $C_T,C_V$ such that $T_*=nL^{C_T}$ bounds the number of
exposures, row replacements, and inverse-column requests in a trial, and
$V=nL^{C_V}$ bounds $\ones^\top\mathbf{R}\ones$.
Scalar Gaussian draws and output coordinates have separate polynomial bounds.
We first round each input entry toward zero on a dyadic grid of spacing
$\Delta_{\mathrm{in}}\le\delta_{\mathrm{in}}/n$, where
$\delta_{\mathrm{in}}=10^{-4}$, obtaining $\bar {\mathbf{A}}$.
Column normalization persists, since
$\sum_i\bar a_{ij}^2\le\sum_i a_{ij}^2\le1$, and every signing satisfies
\begin{align}
 \bigl|((\mathbf{A}-\bar {\mathbf{A}})\varepsilon)_i\bigr|
 &\le\sum_{j=1}^n|a_{ij}-\bar a_{ij}|
 \le n\Delta_{\mathrm{in}}
 \le\delta_{\mathrm{in}}.\label{eq:input-transfer}
\end{align}
This reduction costs $\tO(L_{\mathrm{in}})$ and leaves $O(\log(2n))$ bits
per entry; the row reduction of \Cref{sec:foundations} is applied afterward.
Appendix~\ref{app:precision-penalty} bounds the errors in assembling the
anchor matrices and in smoothing the basis.
 
\begin{proposition}[Finite implementation]
\label{prop:finite-implementation}
{Let the continuous-Gaussian reference make at most
$N_{\mathrm{out}}\le n^{C_o}$ scalar output roundings and stop only on
conditions measurable in its exposed history. Let $G$ be a reference event
on which all Gaussian magnitudes are polynomially bounded and
$\norm{\mathbf{Q}_t}_F+\norm{\mathbf{Q}_t^{-1}}_F\le n^{C_Q}$
at every query and intermediate state, with constant slack in any finite numerical guard.
In explicit batches, also require $\|g\|_2\ge n^{-C_g}$ for a fixed $C_g$.}
 
{\textup{(i)} For any fixed $C_f>0$, inverse-polynomial tolerances
couple the two processes with the same dyadic anchor bases and output rules.
On $G$, all stored constraint rows and directions agree outside an event
of unconditional probability at most $n^{-C_f}$.}
 
{\textup{(ii)} Local errors at most $\epsilon$ cause additional potential
drift and deterministic discrepancy error at most
$T_*\epsilon\operatorname{poly}(n,V,q_*^{-1},\rho^{-1})$.
The reference batch concentration bounds persist after changing absolute constants,
and \textup{(i)} transfers their success probabilities to the finite algorithm.
Conditional estimates concern the reference's original exposed history,
without conditioning on $G$. Each direction and clipping length is exposed before drawing
the independent fair endpoint sign, and each dither is independent of
its scalar conditional on the preceding transcript.}
\end{proposition}
 
Independent dithers make identical stored outputs likely even when queries
are approximate. Agreement keeps errors local: an error in one output is not
fed into later adaptive constraints. Appendix~\ref{app:precision-rounding}
proves this coupling.
 
\paragraph{The rounding rule.}
Choose polynomially small dyadic $\Delta,e$ and a power of
two $M$ with $e\le\Delta n^{-C_f}/(8N_{\mathrm{out}})$ and
$M\ge8N_{\mathrm{out}}n^{C_f}$. We then query and normalize using
\eqref{eq:sampler-output-query} to scalar error at most $e/2$.
For each approximate normalized scalar $\widehat a$, independently draw
uniform $U\in\{0,1/M,\ldots,(M-1)/M\}$ and set
$r=\widehat a+\Delta U$. Abort the trial if
$\operatorname{dist}(r,\Delta\mathbb Z)\le e/2$. Otherwise store
$\Delta\lfloor r/\Delta\rfloor$. From the rounded outputs $\bar h,\widehat c$,
set $\widehat h=3\bar h/4$, with norm in $[1/2,1]$, then zero its fixed
coordinates $J$ as in Algorithm~\ref{alg:batch}. Return $(\widehat h,\widehat c)$.
Store the coloring on a dyadic grid of spacing $\Delta_x$, with fixed
coordinates exactly $\pm1$. Clip using exact cross multiplication and
round the chosen endpoint to this grid inside the cube, preserving attained
endpoints. Let $\epsilon$ be the local error after amplification.\looseness=-1
 
Appendix~\ref{app:precision-robustness} bounds polynomial error amplification
and cumulative drift, preserving the negative barriers and deposit margin.
Fix a total implementation allowance $\zeta=10^{-3}$ and choose the
effective error to satisfy
\begin{equation}
\label{eq:precision-hierarchy}
 \epsilon\le
 \min\left\{
 \frac{\zeta}{CT_*(1+V)^3},\,
 \frac{\zeta\sqrt{q_*}}{CT_*\sqrt n},\,
 \frac{a}{CT_*},\,
 \frac{1}{CT_*n}
 \right\}.
\end{equation}
We choose the anchor, smoothing, evaluation, and quantization tolerances
polynomially smaller still. Each inverse tolerance has logarithm $O(\log n)$,
and stable linear algebra needs $O(\log^2 n)$ guard bits. These choices depend
only on $n$ and universal constants, and make no spectral-gap assumptions.

At an anchor, the construction of \Cref{sec:algorithm} is carried out with the
following changes. Estimate every load $t_i=z_i^\top v$ to error $\Theta/8$
and mark a row heavy when the estimate reaches $3\Theta/4$. Normalize the
gradients that form $\mathbf H$, fix the frozen rows' weights at
$\nu(u_i(x_0)+\mu)$, and recompute the anchor matrices and guards. Here $\mu=Cb_{\max}\epsilon$, where $b_{\max}$ bounds
batch exposures and $\epsilon$ is the amplified local error above. This
buffer makes the fixed weights dominate the true weights despite constraint
errors. The tolerance choices keep every shifted barrier argument negative.
Apply the verified construction of Lemma~\ref{lem:penalty-bound} to
$\mathbf{P}_0-(1-\rho)\mathbf{D}_0$ and $\mathbf H$, with at most
$C_A\log(2n)$ independent attempts. Restart the trial on exhaustion;
otherwise obtain $k=\lceil s/32\rceil$ approximate basis directions. Add the
Gaussian perturbation and dyadic rounding from Section~\ref{sec:sampler} to
obtain $\mathbf{Z}$.
 
At every anchor, interval evaluation checks a dyadic vector $\widehat v$ with
$\widehat v>0$, $(b\mathbf{I}-\mathbf{M})\widehat v\ge\ones/2$ entrywise,
and $\ones^\top\widehat v\le NP/2$. These conditions imply
$\mathbf{M}\prec b\mathbf{I}$ and $\ones^\top\mathbf{R}\ones\le NP$.
The guard costs $\widetilde O(s^2)$ and has constant slack under the
smaller potential cap enforced by increasing $C_P$.
Repeat it after buffering the fixed weights. Appendix~\ref{app:anchor-guards}
gives the comparisons, and completion queries use the sampler's conditioning
guarantee. On the good event, the direction in Algorithm~\ref{alg:batch} has
norm in $[1/2,1]$. We use the inverse-polynomial tolerances of this section
and restart after an undefined operation, overflow, or failed guard.\looseness=-1
 
\begin{proof}[Proof of Theorem~\ref{thm:main}]
Round the input as in \eqref{eq:input-transfer}. Initial row reductions give
$\mathbf{M}\preceq b_{\mathrm{dep}}\mathbf{I}$. At each anchor,
Proposition~\ref{prop:batch-geometry} and
Lemma~\ref{lem:penalty-bound} give
$k=\lceil s/32\rceil$ approximate basis directions, with smoothing and
rounding errors small enough for Proposition~\ref{prop:finite-implementation}.
Near-spherical directions and approximate orthogonality allow the progress,
window, resolvent, and retirement bounds. Choose constants making each of
four groups of bad events have probability at most $1/16$:
{potential and work caps, trust windows and retirement,
verified anchor construction and conditioning, and finite-query disagreement.
For guard passage, use the smaller potential cap from
Appendix~\ref{app:anchor-guards}, obtained by increasing $C_P$.} The polynomial operation
cap permits a union bound in the last two groups, with compatible
inverse-polynomial tolerances. On the joint good event, affine discrepancy
and numerical error give\looseness=-1
\begin{align}
 \norm{\mathbf{A}\eps}_\infty
 &\le\norm{\bar {\mathbf{A}}\eps}_\infty+10^{-4}<\frac{31673}{320}+10^{-2}+10^{-3}+10^{-4}
<98.99<99.
 \label{eq:assembly-discrepancy}
\end{align}
A trial passes exact verification with probability at least $3/4$.
Equation~\eqref{eq:algorithm-verification} guarantees discrepancy below $99$
for every accepted output, including outside this event.
Section~\ref{sec:dynamic} bounds all epochs by
$\widetilde O(n^\omega)$ bit operations after input preprocessing.
Input reduction, incidence updates, and verification add
$\widetilde O(L_{\mathrm{in}})$.
Independent restarts need at most $4/3$ expected trials, proving the bit complexity.\looseness=-1
\end{proof}

\section*{AI Disclosure and Acknowledgements}

We acknowledge an extended conversation with OpenAI's GPT-6 Astra model about the arguments in \cite{Li2026} and possible ways to
batch the walk to exploit recent results in dynamic matrix inversion. We also used Cognition's Devin agent with the SWE-2 model as well as Claude Opus 5.5 to improve the presentation, and in \url{https://github.com/emiletimothy/komlos-omega-lean}, we document our efforts in providing a formal Lean certificate. Finally, we thank Abhishek Shetty, Jan van den Brand, and Ijay Narang for helpful comments. EA's research is supported by NSF Grant CCF 2338816. \looseness=-1

\bibliographystyle{alpha}
\bibliography{main}

\appendix

\section{Row representations and discrepancy reduction}
\label{app:rows}

Write $K=K_{\rm rat}$ in the proofs.
Cauchy-Schwarz gives $1<\norm{a_i}_1^2\le n\norm{a_i}_2^2$ for each remaining row, so column normalization bounds their number by $\sum_{i\text{ remaining}}1<n\sum_{i\text{ remaining}}\norm{a_i}_2^2\le n\sum_{j=1}^n\sum_{i=1}^m a_{ij}^2\le n^2$. For the rational scale, let $a=\max_j|b_{ij}|>0$ and choose the least integer $\ell\ge0$ with $\gamma^{2\ell}a^2\ge q_i$, where $\gamma=65/64$. Set $r_i=\min\{\gamma^\ell a,R_{\max},r_i^{\rm previous}\}$, omitting the previous scale if none exists. Each candidate is at least $\sqrt{q_i}$, and $\gamma^\ell a<\gamma\sqrt{q_i}$ by minimality, including when $\ell=0$. Thus $r_i$ is non-increasing, $q_i\le r_i^2<\gamma^2q_i$, and $r_i\le R_{\max}$. Since $q_i/a^2\le s$, the choice takes $O(\log(2s))$ rational comparisons.
Differentiating \eqref{eq:rational-weight} gives $0<\nu(u)\le119$, $\nu'(u)/\nu(u)=K/(2K-u)$, and $\nu''(u)/\nu'(u)=(K+1)/(2K-u)\le(1+1/K)/2$ for $u\le0$. For every fixed nonnegative integer $r$,
\begin{equation}
\begin{aligned}
\frac{\nu'(u)^2}{\nu(u)}&=\frac K{K+1}\nu''(u),
&\frac{d}{du}\log\nu^{(r)}(u)&=\frac{K+r}{2K-u}.
\end{aligned}
\label{eq:weight-identities}
\end{equation}
Integrating the log derivative shows that changing a non-positive barrier by at most $a\le1$ changes $\nu$, $\nu'$, and $\nu''$ by factors at most $\exp(O(a))$.

\paragraph{Verification of row updates.}
\label{app:row-update-proof}
We prove the row-update guarantees stated in Section~\ref{sec:foundations}.
Initially, and when a large row becomes medium, $d_i=0$. Since $G_i\le q_i\le r_i^2$ and $r_i\le R_{\max}$, both barriers satisfy $u_{i\sigma}\le16-H_0/r_i\le-5864/291<-20$. We show that
\begin{equation}
\sum_{\sigma}\nu(u_{i\sigma})z_iz_i^\top
 \preceq\frac{1000}{1009}\,b_{\mathrm{dep}}
                  \operatorname{diag}(p_i).
\label{eq:birth-deposit}
\end{equation}
For nonnegative $p$ with $q=\sum_jp_j$, Cauchy-Schwarz gives $pp^\top\preceq q\operatorname{diag}(p)$. Together with $q_i\le r_i^2$, this reduces \eqref{eq:birth-deposit} to bounding $f(r)=2\nu(16-H_0/r)/r^2$. For $0<r\le R_{\max}$, $rf'(r)/f(r)=\bigl[(K-2)H_0/r-4K+32\bigr]/\bigl[2K-16+H_0/r\bigr]>0$, so $f(r)\le f(R_{\max})$. The exact comparison
\begin{align}
\left(1+\frac{2932}{291K}\right)^K
\ge\sum_{\ell=0}^{24}\binom{10000}{\ell}
                \left(\frac{2932}{2910000}\right)^\ell\notag>23634 >\frac{1009}{1000}\,
       \frac{5950000000}{254043}=\frac{1009}{1000}
       \frac{2\cdot119}{b_{\mathrm{dep}}R_{\max}^2}.
\label{eq:rational-birth-bound}
\end{align}
proves the margin: multiplying the 25 summands by $10^3$ and summing their floors gives $23634835$. Summing \eqref{eq:birth-deposit} and the large-row deposits, column normalization gives the initialization bound.

Transfers preserve $d_i$. At fixed scale, deletion decreases $G_i$ and $p_i$ coordinatewise, hence decreases the barrier, scalar weight, and rank-one contribution entrywise. For scale reduction, hold $G,d,p$ fixed with $u\le0$ and $G\le q\le r^2$. Differentiation gives
\begin{equation*}
r\frac{\partial u}{\partial r}=12-u-\frac{4G}{r^2}\ge8,
\qquad
r\frac{\partial}{\partial r}\log\frac{\nu(u)}{r^4}
 =\frac{4K(1-G/r^2)+(4-K)u}{2K-u}\ge0.
\end{equation*}
Thus decreasing $r$ keeps $u$ nonpositive and decreases both $u$ and $\nu(u)/r^4$, since $K\ge4$. Equation~\eqref{eq:birth-deposit} handles newly medium rows. Fixed-coordinate transfers preserve energy because $1-x_j^2=0$, and their removal takes a principal submatrix. Retirement removes two PSD, entrywise nonnegative terms.\looseness=-1

\subsection{Final discrepancy}

For a representation change $\varphi_i^+(y)=\varphi_i(y)+\Delta c_i+\Delta b_i^\top y$ at $x^*$, define the anchor jump $\upsilon_i=\varphi_i^+(x^*)-\varphi_i(x^*)$. Each change satisfies $\varphi_i^+(y)-\varphi_i(y)=\upsilon_i+\Delta b_i^\top(y-x^*)$. Telescoping from the initial representation $\varphi_i(y)=a_i^\top y$ over changes indexed by $e$ at anchors $x^e$ gives
\begin{equation}
a_i^\top y-\varphi_i^{\mathrm{final}}(y)
 =\sum_e\left[-\upsilon_{ie}
        -\Delta b_{ie}^\top(y-x^e)\right].
\label{eq:affine-defect}
\end{equation}

A fixed-entry transfer contributes $b_{ij}(y_j-x_j^*)$, which vanishes at later colorings. A deletion has the same form, bounded in magnitude by $2|b_{ij}|$ on the cube. Retirement contributes $(b_i^{\mathrm{ret}})^\top(y-x^{\mathrm{ret}})$, and an offset-only repair contributes $-\upsilon_{ie}$. Moving or rounding $x$ causes no defect at a fixed representation. Adjusting the offset to round the tracked sum creates an anchor jump.

\begin{proposition}[Deterministic discrepancy reduction]
\label{prop:final-defect-bound}
Suppose the walk completes Algorithm~\ref{alg:row-update} with $s<64$. Round remaining coordinates to their nearest signs, choosing $+1$ at zero. Assume every nonretired medium barrier stayed nonpositive. At this signing, suppose retirement defects have magnitude at most $\delta_{\mathrm{ret}}$, each row's total offset-only defect has magnitude at most $\delta_{\mathrm{off}}$, and every row that was never medium had tracked value of magnitude at most $\delta_{\mathrm{large}}$ at its last anchor. Then
\begin{align}
\norm{\mathbf{A}\varepsilon}_\infty
 <\max\left\{
 \frac{31673}{320}+\delta_{\mathrm{ret}}+\delta_{\mathrm{off}},\,
 63+\delta_{\mathrm{large}}+\delta_{\mathrm{off}},\,1
 \right\}.
\label{eq:final-defect-bound}
\end{align}
\end{proposition}
\begin{proof}
Deletion requires $b^2>q/64$, so rationalizing the decrease in row norm gives
\begin{equation*}
\frac{|b|}{\sqrt q-\sqrt{q-b^2}}
 =\frac{\sqrt q}{|b|}+\sqrt{\frac q{b^2}-1}<8+\sqrt{63}.
\end{equation*}
Deletions begin with row norm at most $R_{\max}$, which never increases. Telescoping its decreases bounds the total absolute value of deleted coefficients by $R_{\max}(8+\sqrt{63})$, hence their defect by twice this amount throughout the cube.
No medium row survives the final update: $q_i>0$ and $s<64$ implies $\max_jp_{ij}\ge q_i/s>q_i/64$, contradicting spread. {For the last deletion, use the preceding positive scale: the transfer preserves $d_i$ and decreases $G_i$, so both barriers stay nonpositive. The completed row therefore has a constant representation of magnitude below $H_0$, even though no new scale is needed at $q_i=0$.} Equation~\eqref{eq:affine-defect} gives $|a_i^\top\varepsilon|<H_0+2R_{\max}(8+\sqrt{63})+\delta_{\mathrm{ret}}+\delta_{\mathrm{off}}$. Since $\sqrt{63}<127/16$, $H_0+2R_{\max}(8+\sqrt{63})<263/5+(291/100)(255/16)=31673/320$, proving the first term in \eqref{eq:final-defect-bound}.
A still-large row has undergone no deletion or retirement. Rounding changes at most 63 coordinates by at most one each, giving the second term in \eqref{eq:final-defect-bound} because $|b_{ij}|\le1$. This also covers rows that passed directly from large to completed. Discarded input rows have discrepancy at most one.\looseness=-1
\end{proof}

\section{Curvature and stability}
\label{sec:appendix-geometry}

Use the notation of Section~\ref{sec:curvature} and set $\alpha=4/25$. The row invariants give $0\le z_{ij}\le1/64$, $\gamma^{-2}<\ones^\top z_i\le1$, and $\norm{z_i}^2\le(\ones^\top z_i)^2/64$ for every nonempty medium row.

\begin{lemma}[Positivity, monotonicity, and a ground-state identity]
\label{lem:curv-ground-state}
Suppose $\mathbf{M}\succeq0$ is symmetric and entrywise nonnegative, and
$\mathbf{M}\prec b\mathbf{I}$.  Then $\mathbf{R}$ is entrywise nonnegative, $v\ge b^{-1}\ones$
entrywise, and $F_b(\mathbf{M})\ge0$.  For every $h\in\R^s$,
\begin{equation}
\begin{aligned}
 h^\top\bm{\mathcal A} h
 &=\sum_j\frac{h_j^2}{v_j}
   +\sum_{j<\ell}\mathbf{M}_{j\ell}v_jv_\ell
        \left(\frac{h_j}{v_j}-\frac{h_\ell}{v_\ell}\right)^2,\quad \bm{\mathcal A}\succeq\diag(1/v),
 \quad \norm{\mathbf{R}}_{\op}\le\max_jv_j\le\ones^\top v.
\end{aligned}
 \label{eq:curv-ground-state}
\end{equation}
The potential is nonincreasing when $b$ increases, when $\mathbf{M}$ decreases
in Loewner order or entrywise while remaining nonnegative, or under
principal deletion with the corresponding reduction in $s$.
\end{lemma}

\begin{proof}
The convergent Neumann series gives
\begin{align}
 \mathbf{R}=\sum_{q\ge0}b^{-q-1}\mathbf{M}^q,\qquad
 F_b(\mathbf{M})=\sum_{q\ge1}b^{-q-1}\ones^\top \mathbf{M}^q\ones.
\end{align}
The series gives positivity and monotonicity in $b$ and the entries of $\mathbf{M}$. Principal deletion removes nonnegative terms. Loewner monotonicity follows from
$(b\mathbf{I}-\mathbf{M}_{\rm new})^{-1}\preceq(b\mathbf{I}-\mathbf{M})^{-1}$ whenever $\mathbf{M}_{\rm new}\preceq \mathbf{M}$.
Substituting $bv_j=1+\sum_\ell \mathbf{M}_{j\ell}v_\ell$ into $h^\top\bm{\mathcal A}h$ and pairing the terms indexed by $(j,\ell)$ and $(\ell,j)$ gives \eqref{eq:curv-ground-state}.
Finally, symmetry and nonnegativity give $\norm{\mathbf{R}}_{\op}\le\norm{\mathbf{R}}_\infty=\max_j(\mathbf{R}\ones)_j$.
\end{proof}

\begin{proof}[Proof of Proposition~\ref{prop:batch-geometry}]
For the batch matrix \eqref{eq:batch-fixed-matrix} and curvature matrices
\eqref{eq:batch-curvature}, differentiation along $x+th$ gives
$\mathbf R'=\mathbf R\mathbf M'\mathbf R$, $v'=\mathbf R\mathbf M'v$,
$\mathbf M'[h]v=\mathbf Ch$, and
$v^\top\mathbf M''[h,h]v=\mathbf U[h]-\mathbf D[h]$.
So,
\begin{equation}
 \nabla F_b=\sum_{i\in\mathcal L}\nu'(u_i)t_i^2g_i,
 \qquad \nabla^2F_b=\mathbf U+2\mathbf C^\top\mathbf R\mathbf C-\mathbf D
 =\mathbf P-\mathbf D,
 \label{eq:curv-hessian-potential}
\end{equation}
via $F'=v^\top\mathbf M'v$ and
$F''=v^\top\mathbf M''v+2(\mathbf M'v)^\top\mathbf R(\mathbf M'v)$,
with the barrier derivatives from \eqref{eq:barrier-quadratic}.

\emph{Anchor dimension.}
Following Li~\cite{Li2026}, we bound the dimension loss from the constraints
and the two positive terms of $\mathbf P$.
Large rows cost fewer than $s/R_{\max}^2$ dimensions because
$\sum_\ell q_\ell\le s$. Nearly tight rows cost fewer than
$5\gamma^4s/238$, since
\[
 \nu(-1)>119e^{-1/2}>357/5,
 \qquad
 \sum_i\nu(u_i)(\ones^\top z_i)^2
 =\ones^\top\mathbf W\ones<bs<3s/2,
\]
and each such row contributes more than $(357/5)\gamma^{-4}$.
Heavy rows cost at most $s/2048$ dimensions:
$\sum_i t_i\le\Lambda\ones^\top v\le\Lambda\mathcal P s$ and
$\Theta=4096\Lambda\mathcal P$.
For the existence argument, let $\mathbf J$ project onto the eigenspaces
of $\mathbf M_0$ with eigenvalues at least $\alpha b$.
Because $\mathbf B\preceq b_{\rm dep}\mathbf I$ and the rows are spread,
\begin{align*}
 \rank\mathbf J
 &\le\frac{\tr(\mathbf J\mathbf W_0\mathbf J)}{\alpha b-b_{\rm dep}}
 \le\frac{\tr\mathbf W_0}{\alpha b-b_{\rm dep}}
 \le\frac{\ones^\top\mathbf W_0\ones}{64(\alpha b-b_{\rm dep})}
 <\frac{bs}{64(\alpha b-b_{\rm dep})}
 \le\frac{s}{64(\alpha-b_{\rm dep})}.
\end{align*}
Impose $\mathbf J\mathbf C_0h=0$, which costs at most this many dimensions.
Let $\bm\Gamma=[\sqrt{\nu(u_i)}z_i]_{i\in\mathcal L}$ and
$a_i(h)=\nu'(u_i)t_i(g_i^\top h)/\sqrt{\nu(u_i)}$, at the anchor.
Then $\mathbf C_0h=\bm\Gamma a(h)$,
$\bm\Gamma\bm\Gamma^\top\preceq\mathbf M_0$, and
$\norm{a(h)}^2=K_{\rm rat}\mathbf U_0[h]/(K_{\rm rat}+1)\le\mathbf U_0[h]$.
Conjugating the matrix comparison by
$(\mathbf I-\mathbf J)\mathbf R_0^{1/2}$ gives
\[
 (\mathbf C_0h)^\top\mathbf R_0(\mathbf C_0h)
 \le\frac{\alpha}{1-\alpha}\norm{a(h)}^2,
 \qquad
 \mathbf P_0[h]\le\frac{1+\alpha}{1-\alpha}\mathbf U_0[h].
\]
Also, the weight identities and gradient bound imply
\[
 (\mathbf U_0)_{jj}
 \le4\sum_{i\in\mathcal L}\nu''(u_i)t_i^2z_{ij}
 \le\frac{1+1/K_{\rm rat}}4(\mathbf D_0)_{jj}.
\]
If $\mathbf Q\succeq0$ and diagonal $\mathbf D\succeq0$ satisfy
$\mathbf Q_{jj}\le\theta\mathbf D_{jj}$, the coordinates with
$\mathbf D_{jj}=0$ have zero rows in $\mathbf Q$.
On the other coordinates, congruence by $\mathbf D^{-1/2}$ shows that
$\mathbf Q-\mathbf D$ has at most
$\tr(\mathbf D^{-1/2}\mathbf Q\mathbf D^{-1/2})\le\theta s$
positive eigenvalues. Intersecting its nonpositive spectral subspace
with the previous constraints therefore loses at most $\theta s$
more dimensions. Apply this with
$\mathbf Q=(1+\alpha)\mathbf U_0/(1-\alpha)$ and
$(1-\rho)\mathbf D_0$ in place of $\mathbf D$.
The five losses give a subspace $E\subseteq\ker\mathbf H$ with
$\mathbf P_0[h]\le(1-\rho)\mathbf D_0[h]$ and $\dim E>\mu s$, where\looseness=-1
\begin{equation}
 \mu=1-\frac1{R_{\max}^2}-\frac{5\gamma^4}{238}
 -\frac1{64(\alpha-b_{\rm dep})}-\frac1{2048}
 -\frac{(1+\alpha)(1+1/K_{\rm rat})}{4(1-\alpha)(1-\rho)}
 >\frac1{32}.
 \label{eq:curv-dimension-constant}
\end{equation}
Thus at least $\lceil s/32\rceil$ nonpositive directions exist.
The stability argument below uses only the anchor inequality, so any
such subspace of $\mathbf P_0-(1-\rho)\mathbf D_0$ in $\ker\mathbf H$ works.

\emph{Stability.}
Write $y=x-x_0$ and
$\mathbf M_{\mathcal L,0}=\sum_{i\in\mathcal L}\nu(u_i(x_0))z_iz_i^\top$.
On the window \eqref{eq:batch-window}, fixed rows satisfy
$u_i(x)-u_i(x_0)=-4\sum_jz_{ij}y_j^2\le0$, so their fixed weights
majorize their true weights. Light rows start below $-1$ and change by
at most $a+4a^2\le2a$ when $a\le1/4$, so they stay negative.
By \eqref{eq:weight-identities}, $\nu,\nu',\nu''$ change by factors $1+O(a)$.

Since light loads are at most $\Theta$ and their weights at most $119$,
$\mathbf M_{\mathcal L,0}v_0\le119\Theta\Lambda\ones\le\kappa\ones$
entrywise. Weighted Cauchy--Schwarz gives
$zz^\top\preceq(z^\top v_0)\diag(z/v_0)$; hence
Lemma~\ref{lem:curv-ground-state} implies
\[
 \mathbf M_{\mathcal L,0}
 \preceq\diag\left(\frac{\mathbf M_{\mathcal L,0}v_0}{v_0}\right)
 \preceq\kappa\diag(1/v_0)\preceq\kappa\bm{\mathcal A}_0.
\]
For $\bm\Delta=\mathbf M_x-\mathbf M_0$,
relative light-weight changes give both
$-Ca\mathbf M_{\mathcal L,0}\preceq\bm\Delta\preceq Ca\mathbf M_{\mathcal L,0}$
and $|\bm\Delta|\le Ca\mathbf M_{\mathcal L,0}$ entrywise.
The first comparison keeps $\bm{\mathcal A}_x$ positive definite when
$Ca\kappa<1$ and bounds its inverse in Loewner order.
The second, with $\mathbf R_x\ge0$ entrywise, gives
\[
 |v_x-v_0|=|\mathbf R_x\bm\Delta v_0|
 \le Ca\kappa\mathbf R_x\ones=Ca\kappa v_x.
\]
Thus, with universal constants,
\begin{equation}
 \mathbf R_x\preceq(1+Ca\kappa)\mathbf R_0,
 \quad t_{ix}=(1+O(a\kappa))t_{i0},
 \quad(1-Ca\kappa)\mathbf D_0\preceq\mathbf D_x
       \preceq(1+Ca\kappa)\mathbf D_0.
 \label{eq:curv-D-stability}
\end{equation}
These estimates also justify continuation from $x_0$ along every segment
in the convex window: the slack stays uniformly positive before any
possible first exit.
Fix $h\in E$. Since $\ones^\top z_i\le1$ and $\norm y_\infty\le a$,
\begin{align}
 ((g_i(x)-g_i(x_0))^\top h)^2
 \le64a^2\sum_jz_{ij}h_j^2,\quad
 \sum_{i\in\mathcal L}\nu''(u_{i0})t_{i0}^2
       ((g_i(x)-g_i(x_0))^\top h)^2
 &\le5a^2\mathbf D_0[h].
 \label{eq:curv-gradient-motion}
\end{align}
The triangle inequality for the vector
$(\sqrt{\nu''(u_i)}t_i g_i^\top h)_{i\in\mathcal L}$ gives
\begin{equation}
 \sqrt{\mathbf U_x[h]}
 \le(1+Ca\kappa)\bigl(\sqrt{\mathbf U_0[h]}
                         +3a\sqrt{\mathbf D_0[h]}\bigr).
 \label{eq:curv-U-stability}
\end{equation}
For $\bm\Gamma_0=[\sqrt{\nu(u_{i0})}z_i]_{i\in\mathcal L}$, slack domination
also gives $\norm{\mathbf R_0^{1/2}\bm\Gamma_0}_{\op}^2\le\kappa$.
Let $(\mathbf C_x-\mathbf C_0)h=\bm\Gamma_0 e$, where\looseness=-1
\[
 e_i=\frac{\nu'(u_{ix})t_{ix}-\nu'(u_{i0})t_{i0}}{\sqrt{\nu(u_{i0})}}
                  g_i(x_0)^\top h
     +\frac{\nu'(u_{ix})t_{ix}}{\sqrt{\nu(u_{i0})}}
                  (g_i(x)-g_i(x_0))^\top h.
\]
The weight identities and \eqref{eq:curv-gradient-motion} bound
$\norm e\le Ca\kappa\sqrt{\mathbf U_0[h]}+Ca\sqrt{\mathbf D_0[h]}$,
so
\begin{equation}
 \norm{\mathbf R_0^{1/2}(\mathbf C_x-\mathbf C_0)h}
 \le Ca\sqrt\kappa\bigl(\kappa\sqrt{\mathbf U_0[h]}
                              +\sqrt{\mathbf D_0[h]}\bigr).
 \label{eq:curv-C-stability}
\end{equation}
The triangle inequality for
$(\sqrt{\mathbf U_x[h]},\sqrt2\norm{\mathbf R_x^{1/2}\mathbf C_xh})$,
together with $\mathbf U_0[h]\le\mathbf P_0[h]\le\mathbf D_0[h]$, gives
\begin{align}
 \sqrt{\mathbf P_x[h]}
 &\le\sqrt{\mathbf P_0[h]}
       +Ca(1+\kappa)^{3/2}\sqrt{\mathbf D_0[h]},
 \label{eq:curv-P-root-stability}\\
 \mathbf P_x[h]
 &\le\mathbf P_0[h]+Ca(1+\kappa)^2\mathbf D_0[h].
 \label{eq:curv-P-stability}
\end{align}
Here $C$ absorbs the multiplicative inverse error, and squaring uses
$a(1+\kappa)^2$ sufficiently small.
Together with \eqref{eq:curv-D-stability} and the anchor margin,
\[
 \mathbf P_x[h]-(1-\rho/2)\mathbf D_x[h]
 \le[-\rho/2+Ca(1+\kappa)^2]\mathbf D_0[h]\le0
\]
when $a\le c\rho/(1+\kappa)^2$ for a sufficiently small universal $c$.
The same estimates cover $\mathbf D_0[h]=0$, in which case all these
quadratic forms vanish.
The Hessian identity proves concavity. Midpoint concavity gives
$\E_\sigma F_b(\mathbf M(x+\sigma th))\le F_b(\mathbf M(x))$ whenever
both endpoints are in the window and the fair sign is chosen after $t$.
\end{proof}

\emph{Buffered weights.}
The proof also allows any fixed weight $\beta_i\ge\nu(u_i(x_0))$,
including $\beta_i=\nu(u_i(x_0)+\mu)$ when $u_i(x_0)+\mu\le0$.
The nearly tight count needs only this lower bound, the trace comparison
uses nonnegative coefficients, and the light-row Gram matrix stays below
the anchor matrix. The finite-bit analysis bounds the buffer's potential cost.

\emph{Additive anchor error.}
Suppose a fixed subspace $\widetilde E$ instead satisfies
$\mathbf P_0[h]\le(1-\rho)\mathbf D_0[h]+\xi\norm h^2$, with $\xi\ge0$.
Set $\widehat{\mathbf D}_x=\mathbf D_x+(2\xi/\rho)\mathbf I$.
Then $\mathbf P_0\preceq(1-\rho/2)\widehat{\mathbf D}_0$ on
$\widetilde E$, while $\widehat{\mathbf D}_x\succeq\mathbf D_x$ and
$\widehat{\mathbf D}_x=(1+O(a\kappa))\widehat{\mathbf D}_0$ in Loewner order.
Repeating the stability argument with margin $\rho/2$ gives for all $h\in \tilde{E}$,\looseness=-1
\begin{equation}
 \begin{aligned}
 h^\top\nabla^2F_b(\mathbf M(x))h
 \le(1-\rho/4)\widehat{\mathbf D}_x[h]-\mathbf D_x[h]\le-\frac\rho4\mathbf D_x[h]+\frac{2\xi}{\rho}\norm h^2
 \end{aligned}
 \label{eq:curv-additive-stability}
\end{equation}

\noindent\emph{Anchor guards.}\label{app:anchor-guards}
The guards in Section~\ref{sec:precision} are checked by interval products
in $\widetilde O(s^2)$ work. Since $\widehat v>0$ and
$\mathbf M\widehat v<b\widehat v$, diagonal similarity by
$\diag(\widehat v)$ gives row sums below $b$, hence
$\lambda_{\max}(\mathbf M)<b$. Multiplying
$(b\mathbf I-\mathbf M)\widehat v\ge\ones/2$ by the nonnegative inverse gives
$\mathbf R\ones\le2\widehat v$ and $\ones^\top\mathbf R\ones\le NP$.
If $F\le NP/4$ and $P\ge8$, an accurate dyadic approximation to
$\mathbf R\ones$ passes with constant slack. Increasing $C_P$ enforces
this smaller cap with the same constant success probability.
Repeat the guard after buffering the fixed weights.
\section{Probability bounds for the capped walk}
\label{app:walk}

\begin{proof}[Proof of Proposition~\ref{prop:walk-success}]
Use the stopped process from Section~\ref{sec:walk}. Directions and clipping
lengths are exposed before drawing the independent fair endpoint signs.
For each batch, the sampler and rounding estimates in
Sections~\ref{sec:sampler} and~\ref{sec:precision} give
\begin{equation}
 \left\|[h_1\ \cdots\ h_b]^{\mathsf T}
              [h_1\ \cdots\ h_b]\right\|_{\mathrm{op}}\le2.
 \label{eq:walk-gram}
\end{equation}

\noindent\emph{A common concentration estimate.}
Let $a_t$ be exposed before the fair sign $\sigma_t$, and suppose
$\sum_ta_t^2\le Q$ on the stopped process. Since
$\mathbb E[e^{\lambda\sigma_ta_t}\mid\mathcal F_t,a_t]
=\cosh(\lambda a_t)\le e^{\lambda^2a_t^2/2}$,
the exponential supermartingale and its version with $-\lambda$ give
\begin{equation}
 \Pr\left[\max_\ell\left|\sum_{t\le\ell}\sigma_ta_t\right|\ge u\right]
 \le2\inf_{\lambda>0}e^{-\lambda u+\lambda^2Q/2}
 =2e^{-u^2/(2Q)}.
 \label{eq:walk-maximal}
\end{equation}
The same estimate holds conditional on any starting history.

\noindent\emph{Trust windows.}
Fix a vector $f$ at the batch anchor and take $a_t=t_tf^{\mathsf T}h_t$.
Equation~\eqref{eq:walk-gram} gives
$\sum_ta_t^2\le2\tau^2\|f\|^2$.
Apply \eqref{eq:walk-maximal} to the $s$ coordinate normals and at most
$2s/q_*$ signed barrier gradients, whose norms are at most two.
With $u=a/4$ and $\tau\le a/(2^{20}L)$, each failure probability is at most
$\exp(-c2^{40}L^2)$. There are polynomially many forms in a capped trial,
so all batches stay in their inner windows with probability at least
$1-n^{-8}$, after choosing the absolute constants.
The quadratic identity \eqref{eq:barrier-quadratic} bounds barrier motion by
$a/4+a^2/4$. Fixed rows never exceed their anchor barriers, and light rows
start below $-1$ and stay below $-3/4$. Both proposed endpoints remain in
the outer window, where Proposition~\ref{prop:batch-geometry} applies.

\noindent\emph{The potential and resolvent.}
Within a valid window, concavity gives
$\mathbb E[F(x_{t+1})\mid\mathcal F_t,h_t]\le F(x_t)$.
Lemma~\ref{lem:curv-ground-state} and the
\hyperref[app:row-update-proof]{row-update comparisons} in
Appendix~\ref{app:rows} extend this inequality across batch boundaries
and epoch barrier increases. Stop at the first window or cap failure.
At the initial anchor, $\mathbf M\preceq b_{\rm dep}\mathbf I$ gives
$F_1(\mathbf M)\le Nb_{\rm dep}/(1-b_{\rm dep})\le N$.
At epoch $j\ge1$, $\mathbf M\prec b_{j-1}\mathbf I$ and
$b_j-b_{j-1}=1/(4j^2)$ give
$(b_j\mathbf I-\mathbf M)^{-1}\preceq4j^2\mathbf I$ and
$F_{b_j}(\mathbf M)\le4Nj^2$ at its start. Doob's inequality yields
\begin{align*}
 \Pr[\text{an epoch violates its cap}]
 &\!\le\!\sum_{j=0}^{\lceil\log_2n\rceil}
       \Pr\left[\sup_{t\in \text{ epoch }j}F_{b_j}(\mathbf M_t)\!>\!NP\right]\!\le\!\frac4P\sum_{j=0}^{\lceil\log_2n\rceil}(j+1)^2
 \!\le\!\frac{CL^3}{P}\!\le\!\frac1{100}.
\end{align*}
For the batch's ambient dimension $s_{\rm amb}$, a query state has
$N/2\le s_t\le s_{\rm amb}\le N$. Hence
\begin{equation}
 F_{b_j}(\mathbf M_t)\le NP,\qquad
 \ones^{\mathsf T}(b_j\mathbf I-\mathbf M_t)^{-1}\ones
 =F_{b_j}(\mathbf M_t)+s_{\rm amb}/b_j
 \le N(P+1)\le\mathcal P s_t.
 \label{eq:walk-resolvent-cap}
\end{equation}
The potential bound includes the terminal endpoint before row updates.
Also $b_j\le1+\pi^2/24<3/2$.

\noindent\emph{Retirement.}
Condition on a row's retirement history. Its remaining coefficients $r$
satisfy $\|r\|^2\le q_*$, and retirement occurs between batches.
Within each later batch $\beta$, restrict $r$ to its ambient coordinates,
giving $r_\beta$ of norm at most $\sqrt{q_*}$. Directions vanish on newly
fixed coordinates, so the later active restrictions $r_t$ satisfy
$r_t^{\mathsf T}h_t=r_\beta^{\mathsf T}h_t$.
The Gram bound and \eqref{eq:walk-batches} therefore imply
\[
 Q:=\sum_{t\text{ after retirement}}t_t^2(r_t^{\mathsf T}h_t)^2
 \le2\tau^2q_*B_{\rm tot}\le C_0q_*L^3.
\]
Apply \eqref{eq:walk-maximal} to
$Y=\sum_{t\text{ after retire}}\sigma_tt_tr_t^{\mathsf T}h_t$,
stopping at a proposal violates the Gram or cap bounds:\looseness=-1
\begin{equation}
 \Pr[|Y|\ge u\mid\mathcal F_{\rm retirement}]
 \le2\exp\left(-\frac{u^2}{2C_0q_*L^3}\right)
 =2\exp\left(-\frac{u^2L}{2C_0c_*}\right).
 \label{eq:walk-retirement-tail}
\end{equation}
Nearest-sign rounding of the last fewer than $64$ coordinates adds at most
$\sqrt{63q_*}<8\sqrt{c_*}L^{-2}$.
Take $u=1/200$ and choose $c_*$ so this rounding contributes at most $1/400$
and the displayed probability is at most $n^{-9}$.
Over fewer than $n^2$ rows, the failure probability is at most $n^{-7}$.
The stochastic error and final rounding use at most $3/400$, leaving
$1/400$ for the drift in Section~\ref{sec:precision}.
Combining the three failure probabilities proves the proposition.
\end{proof}
\section{Adaptive Gaussian sampling proofs}
\label{app:sampler}

Use the completion protocol and exposed history from Section~\ref{sec:sampler},
stopping before the first dependent real insertion.
A nonnegative $\phi$ on PSD matrices is \emph{spectral increasing} if it
depends only on eigenvalues and $\phi(\mathbf{A})\le\phi(\mathbf{B})$
whenever $\mathbf{A}\preceq\mathbf{B}$. Normalizing constants below are finite.

\begin{lemma}[Adaptive radial law and regularity]\label{lem:sampler-elliptical}
At initialization or immediately after a real insertion, let $d=\dim F$.
In orthonormal coordinates, the projected hidden rows
$\mathbf{X}\in\R^{(d-1)\times d}$ have conditional density
\begin{equation}\label{eq:sampler-elliptical-density}
 p(\mathbf{X})=\frac1{\mathcal Z}
 \exp\left(-\frac12\operatorname{tr}(\mathbf{X}\bm{\Sigma}^{-1}\mathbf{X}^\top)\right)
 \phi(\mathbf{X}\bm{\Sigma}^{-1}\mathbf{X}^\top),
\end{equation}
where $\bm{\Sigma}$ is PD and $\phi$ is spectral increasing.
Initially $\bm{\Sigma}=\mathbf{I}$. Predetermined hidden-row deletion and
predictable coordinate restriction preserve this class and its extreme
scale bounds. Exposing the null line and inserting a self-normal whose
projection into $F$ has angle $\theta<\pi/2$ to that line preserve it, with\looseness=-1
\begin{equation}\label{eq:sampler-covariance-recursion}
 \lambda_{\min}(\bm{\Sigma}_{\mathrm{new}})
 \ge\lambda_{\min}(\bm{\Sigma})\cos^2\theta,
 \qquad
 \lambda_{\max}(\bm{\Sigma}_{\mathrm{new}})
 \le\lambda_{\max}(\bm{\Sigma}).
\end{equation}
After at most $b$ self-normal insertions with error at most $2\epsilon_s$,
$Cb\epsilon_s^2\le1/4$ ensures
$e^{-Cb\epsilon_s^2}\mathbf{I}\preceq\bm{\Sigma}\preceq\mathbf{I}$ and
$\kappa(\bm{\Sigma})\le2$.
The unit null vector has representation
\begin{equation}\label{eq:sampler-angular-representation}
 h=\frac{\bm{\Sigma}^{-1/2}u}{\lVert\bm{\Sigma}^{-1/2}u\rVert_2},
 \qquad u\sim\mathrm{Unif}(\mathbb S^{d-1}),
\end{equation}
with a fair orientation, and satisfies \eqref{eq:sampler-subgaussian-square}.
For $m=d-1\ge1$ and $\lambda_{\min}(\bm{\Sigma})\ge a_0>0$,
$\Pr[\sigma_{\min}(\mathbf{X})\le t\mid\mathcal F]\le mt^2/(2a_0)$
for every $t\ge0$.
\end{lemma}
\begin{proof}
\emph{Preserving the radial law.}
Whiten by $\mathbf{Y}=\mathbf{X}\bm{\Sigma}^{-1/2}$, whose density is
proportional to $e^{-\operatorname{tr}(\mathbf{Y}\mathbf{Y}^\top)/2}\phi(\mathbf{Y}\mathbf{Y}^\top)$.
For a matrix with $m$ rows, extend $\phi$ to its column Gram by evaluating
it on the $m$ largest eigenvalues, including zeros. This agrees with the
row Gram. Row deletion adds $zz^\top$ to the surviving column Gram,
and column restriction adds $zz^\top$ to the surviving row Gram.
For either Gram $\mathbf{S}\in\R^{\ell\times\ell}$, the radial factor becomes
\begin{equation}\label{eq:sampler-column-integration}
 \phi_1(\mathbf{S})=\int_{\R^\ell}
       e^{-\lVert z\rVert_2^2/2}\phi(\mathbf{S}+zz^\top)\,dz.
\end{equation}
Rotational invariance and pointwise comparison make $\phi_1$ spectral
increasing. The deletion and restriction are predictable.
For $\mathbf{L}$ with orthonormal columns,
$\bm{\Sigma}^{1/2}\mathbf{L}(\mathbf{L}^\top\bm{\Sigma}\mathbf{L})^{-1/2}$
has orthonormal columns. The same calculation for $\mathbf{X}\mathbf{L}$
gives scale $\mathbf{L}^\top\bm{\Sigma}\mathbf{L}$ and preserves the extreme
eigenvalue bounds. For $m=d-1$, thin QR gives volume factor
$\prod_i\mathbf{T}_{ii}^{d-i}$ and the Gram map has Jacobian
$2^m\prod_i\mathbf{T}_{ii}^{m+1-i}$. Thus $\mathbf{S}=\mathbf{Y}\mathbf{Y}^\top$
has density $(\mathcal Z')^{-1}e^{-\operatorname{tr}(\mathbf{S})/2}\phi(\mathbf{S})$
on the positive definite cone: its determinant exponent is $(d-m-1)/2=0$.
Right orthogonal invariance makes the null line uniform and independent of
$\mathbf{S}$.
With an orthonormal basis $\mathbf{U}$ of its orthogonal complement, write
$\mathbf{Y}=\mathbf{A}\mathbf{U}^\top$, where $\mathbf{A}\in\R^{m\times m}$. Then
\begin{equation}\label{eq:sampler-coarea}
 p(\mathbf{A}\mid\ker \mathbf{Y})
 \propto e^{-\operatorname{tr}(\mathbf{A}\mathbf{A}^\top)/2}
         \phi(\mathbf{A}\mathbf{A}^\top)\det(\mathbf{A}\mathbf{A}^\top)^{1/2}.
\end{equation}
The square Gram map has determinant exponent $-1/2$, explaining the
increasing factor $\det(\mathbf{S})^{1/2}$ in \eqref{eq:sampler-coarea}.
For the physical unit null vector $h$, the whitened null vector is
proportional to $\bm{\Sigma}^{1/2}h$. The scale $\mathbf{K}_h$
has support $h^\perp$ and the Schur complement gives  inverse:\looseness=-1
\begin{equation}
\label{eq:sampler-conditioned-covariance}
 \begin{aligned}
 \mathbf{K}_h=\bm{\Sigma}-
 \frac{\bm{\Sigma} hh^\top\bm{\Sigma}}{h^\top\bm{\Sigma} h},\quad
 (\mathbf{K}_h|_{h^\perp})^{-1}
=\mathbf{P}_{h^\perp}\bm{\Sigma}^{-1}\mathbf{P}_{h^\perp}|_{h^\perp}.
 \end{aligned}
\end{equation}
Its nonzero eigenvalues lie between the extreme eigenvalues of $\bm{\Sigma}$.
Normalize the self-normal's projection into $F$ to $\widetilde h$.
Projection from $h^\perp$ to $\widetilde h^\perp$ has singular values
$1,\ldots,1,|\cos\theta|$ and preserves the whitened radial
factor, transforming the scale by congruence, giving
\eqref{eq:sampler-covariance-recursion}. Deleting the next hidden row restores
the rectangular dimensions. Initially $\bm{\Sigma}=\mathbf{I}$ and $\phi=1$.
The self-normal precedes coordinate constraints, so this comparison applies
in $h$'s sampling space and completes the induction.

\emph{Angular bounds.}
Projection of the self-normal into $F$, followed by normalization,
gives $\theta\le C\epsilon_s$. Thus the recursion yields
$\lambda_{\min}(\bm{\Sigma})\ge(1-C\epsilon_s^2)^b
\ge e^{-C'b\epsilon_s^2}$ and $\lambda_{\max}(\bm{\Sigma})\le1$;
coordinate restrictions preserve these bounds.
Right invariance gives \eqref{eq:sampler-angular-representation}, and the
independent symmetric reference row chooses a fair orientation.
For $\kappa(\bm{\Sigma})\le2$, the denominator bound gives
$|\langle v,h\rangle|^2\le2\lVert v\rVert_2^2U_1^2$ for a sphere
coordinate $U_1$. Its even moments satisfy
$d^q\E U_1^{2q}\le(2q-1)!!$, so
\begin{align*}
 \E e^{dU_1^2/8}
 &\le\sum_{q\ge0}\frac{(2q-1)!!}{8^q q!}
 =\left(1-\frac14\right)^{-1/2}<2.
\end{align*}
The same denominator bound and $\E uu^\top=\mathbf{I}/d$ give the covariance
estimate in \eqref{eq:sampler-subgaussian-square}.

\emph{The hidden-matrix lower tail.}
Translation by $q\mathbf{I}$ in the Gram density above and monotonicity give
\begin{align*}
 \Pr[\mathbf{S}\succeq q\mathbf{I}\mid\mathcal F]
 &=\frac{e^{-mq/2}}{\mathcal Z'}
   \int_{\mathbf{T}\succ0}e^{-\operatorname{tr}(\mathbf{T})/2}
        \phi(\mathbf{T}+q\mathbf{I})\,d\mathbf{T}
 \ge e^{-mq/2}.
\end{align*}
Since $\mathbf{X}\mathbf{X}^\top\succeq a_0\mathbf{S}$, take $q=t^2/a_0$
to obtain the claimed bound from
$1-e^{-mt^2/(2a_0)}\le mt^2/(2a_0)$.
\end{proof}

\begin{proof}[Proof of Theorem~\ref{thm:sampler-conditioning}]
Lemma~\ref{lem:sampler-elliptical} gives the conditional direction bounds.
For the high-probability bounds, choose $A'>A$ so that $n^{-A'}$ per batch
sums to at most $n^{-A}/4$. Choose $C_b,C_k$ sufficiently large depending on
$A',C_\eta,c_0$ and $b\epsilon_s$ sufficiently small.
Nonexplicit batches have $k\ge C_k\log(2n)$.

\emph{Real constraints: coordinate rows.}
For fixed $I$ of size $q\le b\le k/2$, let $\mathbf{U}$ have orthonormal
columns spanning $\ker\mathbf{Z}_{0,I,:}$, of dimension $d\ge k-q$.
Then $(\mathbf{Z}_0+\eta\mathbf{G})_{I,:}\mathbf{U}=\eta\bm{\Gamma}$,
where $\bm{\Gamma}\in\R^{q\times d}$ has independent standard Gaussian entries.
$\chi^2$ bounds on a $1/4$-net of at most $9^q$ points put every
$\|a^\top\bm{\Gamma}\|_2$ in $[3\sqrt d/4,3\sqrt d/2]$,
except with probability $2\cdot9^q e^{-c_1d}$.
Net extension gives $\sigma_{\min}(\bm{\Gamma})\ge\sqrt d/4$.
Right projection decreases the row Gram matrix, so rounding gives
$\sigma_{\min}(\mathbf{Z}_{I,:})
\ge\eta\sigma_{\min}(\bm{\Gamma})-\eta\sqrt k/32\ge\eta\sqrt k/8$.
Gaussian norm bounds ensure \eqref{eq:sampler-basis-conditioning}
except with probability $n^{-A'}/3$.

\emph{Separation from self-normals.}
Condition only on the smoothed basis.
Let $\mathbf{H}$ collect the ideal directions whose self-normals were
inserted, and $\mathbf{N}$ those physical normals.
Later directions are orthogonal to earlier self-normals, so
$|\langle h_i,h_j\rangle|\le2\epsilon_s$ for $i\ne j$ and
\begin{equation}\label{eq:sampler-self-gram}
 \lVert\mathbf{H}\mathbf{H}^\top-\mathbf{I}\rVert_{\mathrm{op}}\le2b\epsilon_s,
 \qquad
 \lVert\mathbf{N}-\mathbf{H}\rVert_{\mathrm{op}}\le2\sqrt b\epsilon_s.
\end{equation}
Let $h_t$ be the $t$th query direction, or zero if it never occurs.
Occurrence is predictable, and a query without insertion ends the batch,
so the first $b$ queries include every inserted self-normal.
At every actual query, $\dim F\ge k-b\ge k/2$.
For fixed unit $v\in E$, \eqref{eq:sampler-subgaussian-square} and the tower
property give, inductively,
$\E\exp((k/64)\sum_{t\le b}\langle v,h_t\rangle^2)\le2^b$.
Markov's inequality then gives
\begin{equation}\label{eq:sampler-angle-tail}
 \Pr\left[\sum_{t\le b}\langle v,h_t\rangle^2>\frac1{64}\right]
 \le2^b e^{-k/4096}.
\end{equation}
This also bounds every prefix sum.
For $V_I=\operatorname{span}\{\mathbf{P}_Ee_j:j\in I\}$, take a $1/4$-net
of at most $9^{|I|}$ unit vectors.
Summing over coordinate sets bounds the coordinate and angular failure
probabilities, respectively, by $2b(9es)^b e^{-c_1k/2}\le n^{-A'}/3$ and
$b(9es)^b2^b e^{-k/4096}\le n^{-A'}/3$.
Net extension gives $\lVert\mathbf{H}\mathbf{P}_{V_I}\rVert_{\mathrm{op}}\le1/6$
for every $I$ and prefix. With $S_H=\operatorname{rowspan}(\mathbf{N})$,
\begin{equation}\label{eq:sampler-angle-bound}
 \lVert\mathbf{P}_{S_H}\mathbf{P}_{V_I}\rVert_{\mathrm{op}}
 \le\frac{1/6+2\sqrt b\epsilon_s}
 {\sqrt{1-2b\epsilon_s}-2\sqrt b\epsilon_s}\le\frac12.
\end{equation}

\emph{Combining the constraints.}
In orthonormal coordinates on $E$, let $\mathbf{D}_I$ have rows
$\mathbf{P}_Ee_j$ for $j\in I$. Coordinate conditioning gives
$\sigma_{\min}(\mathbf{D}_I)
\ge\sigma_{\min}(\mathbf{Z}_{I,:})/\sigma_{\max}(\mathbf{Z})
\ge c'\eta\sqrt{k}$, while \eqref{eq:sampler-self-gram} gives a constant
lower bound for $\sigma_{\min}(\mathbf{N})$.
Separation implies
\begin{align*}
 \lVert\mathbf{N}^\top a+\mathbf{D}_I^\top b\rVert_2^2
 &\ge\tfrac12\bigl(\lVert\mathbf{N}^\top a\rVert_2^2+
                       \lVert\mathbf{D}_I^\top b\rVert_2^2\bigr).
\end{align*}
The combined smallest singular value is at least
$2^{-1/2}\min\{\sigma_{\min}(\mathbf{N}),\sigma_{\min}(\mathbf{D}_I)\}$,
omitting empty blocks. Thus no coordinate insertion is dependent.
Returning to coefficient coordinates changes the bound by at most two,
so every nonempty real constraint matrix has smallest singular value at
least $\beta=n^{-C_\beta}$, for a fixed constant. The three failure bounds
sum to $n^{-A'}$ per batch.

\emph{Completion matrices.}
In orthonormal coordinates on $E$,
collect the real rows in $\mathbf{C}$
and hidden rows in $\mathbf{G}$. Decomposing $E$ into $\mathbf{C}$'s row span
and kernel gives
\begin{align*}
 \mathbf{P}=\begin{pmatrix}\mathbf{C}\\\mathbf{G}\end{pmatrix}
  =\begin{pmatrix}\mathbf{C}_R&0\\\mathbf{G}_R&\mathbf{X}\end{pmatrix}.
\end{align*}
$\mathbf{X}$ has one fewer row than column, and
the preceding bound gives invertible $\mathbf{C}_R$ with
$\lVert\mathbf{C}_R^{-1}\rVert\le\beta^{-1}$.
In the original exposed filtration, Lemma~\ref{lem:sampler-elliptical} gives
$\sigma_{\min}(\mathbf{X})\ge\alpha=n^{-C_\alpha}$ across polynomially many
states with probability at least $1-n^{-A-2}$ for large $C_\alpha$,
before intersecting with the real-constraint event.
Each hidden row's norm is bounded by its original ambient Gaussian,
so Gaussian norm bounds give $\lVert\mathbf{G}\rVert_F\le L=n^{C_L}$ across all batches.
For right-hand side $(a,b)$, solve
$u=\mathbf{C}_R^{-1}a$ and $v=\mathbf{X}^\dagger(b-\mathbf{G}_Ru)$.
This right inverse has norm at most
$D=\beta^{-1}+\alpha^{-1}(1+L\beta^{-1})$, so
$\lVert\mathbf{P}^\dagger\rVert\le D$ and $\dim\ker\mathbf{P}=1$.
Between exposure and self-normal insertion the prefix's previous bound
applies. Self-normal insertion restores the rectangular
density, and coordinate insertions preserve it. Coefficient coordinates
change bounds by constant factors.
For unit $w\in\ker\mathbf{P}$ and independent Gaussian $r$,
$r^\top w\sim N(0,1)$ conditionally on $\mathbf{P}$, and
\begin{align*}
 \mathbf{Q}^{-1}\binom{a}{t}
 =\mathbf{P}^\dagger a+
 w\frac{t-r^\top \mathbf{P}^\dagger a}{r^\top w},\quad
 \lVert \mathbf{Q}^{-1}\rVert_{\mathrm{op}}
 \le D+\frac{1+\lVert r\rVert_2D}{|r^\top w|}.
\end{align*}
The refreshed reference row is independent of both the pre-insertion and
post-insertion prefixes. Gaussian norm bounds and
$\Pr[|r^\top w|\le u\mid\mathbf{P}]\le\sqrt{2/\pi}\,u$
give $\lVert r\rVert_2\le n^{C_r}$ and $|r^\top w|\ge n^{-C_r}$
across polynomially many operations with the required probability.
The row norms bound $\lVert\mathbf{Q}\rVert_F$; the inverse bound converts
to Frobenius norm.\looseness=-1
\end{proof}
\section{Matrix formula maintenance}
\label{app:dynamic}

\begin{proof}[Proof of Theorem~\ref{thm:dynamic-maintenance}]
{Apply the row-query structure to
$\mathbf{A}=\bm{\mathcal A}(\mathbf{Q})^\top$, of order $N=O(s)$,
with $L=\lceil\log_2N\rceil$ levels. Level $\ell$ rebuilds every
$d=2^\ell$ operations, and the top level recomputes the full inverse
every $\Theta(N)$ operations. Perform simultaneous rebuilds from the
top down. At a level-$\ell$ rebuild at time $t$, its parent snapshot
has time $u=t$ or $u=t-d$. Store the inverse rows indexed by
$\mathcal I(t,2d)$ together with the columns changed between $u$ and $t$.
The parent stores $\mathcal I(u,4d)$, which contains the child candidates
by nesting and every changed column by the lookahead property.
Thus every required parent row is available, and each level stores
only $O(d)$ rows. At the top level all rows are available.}
{For parent snapshot $\mathbf{B}$ and child
$\mathbf{B}'=\mathbf{B}+\mathbf{U}\mathbf{V}^\top$, combine repeated updates
into at most $2^\ell$ changed columns $J$. Their coordinate vectors form
$\mathbf{V}$, and their net differences form $\mathbf{U}$.
The small matrix $\mathbf{D}=\mathbf{I}+\mathbf{V}^\top\mathbf{B}^{-1}\mathbf{U}$
governs the inverse correction. For $I$ containing $J$ and the child
candidate set, Woodbury gives}\looseness=-1
\begin{equation}\label{eq:dynamic-partial-woodbury}
 (\mathbf{B}'^{-1})_{I,:}
 = (\mathbf{B}^{-1})_{I,:}-(\mathbf{B}^{-1})_{I,:}\mathbf{U}\mathbf{D}^{-1}(\mathbf{B}^{-1})_{J,:}.
\end{equation}
{Let $\mathsf M(a,b,c)$ be the cost of multiplying $a\times b$ and
$b\times c$ matrices. Splitting the long dimensions into square blocks,
with $d=c2^\ell$ for constant $c$, gives the recomputation and amortized costs}
\begin{equation}
\label{eq:dynamic-level-cost}
\begin{aligned}
 O\left(\mathsf M(d,N,d)+\mathsf M(d,d,N)+d^\omega\right)
 &=O\left(N(2^\ell)^{\omega-1}\right),\\
 \sum_{\ell=0}^L\frac{O(N(2^\ell)^{\omega-1})}{2^\ell}
 &=O\left(N\sum_{\ell=0}^L(2^\ell)^{\omega-2}\right)
 =\tO(N^{\omega-1}).
\end{aligned}
\end{equation}
Initialization costs $O(N^\omega)$.
Let $\mathcal K$ be a sufficiently large fixed power of
$N(1+K)(1+K_C)$ bounding endpoint matrices and update factors. It also
bounds the small inverse, since
$\mathbf{D}^{-1}
 =\mathbf{I}-\mathbf{V}^\top \mathbf{B}'^{-1}\mathbf{U}$ and
$\|\mathbf{D}^{-1}\|_F\le\sqrt{|J|}+\|\mathbf{V}\|_F
\|\mathbf{B}'^{-1}\|_F\|\mathbf{U}\|_F$.

{Although \cite[Theorem A.1(4)]{ABGZ2024} assumes an offline
index sequence, its Woodbury error analysis depends only on the norm
bounds, not on advance knowledge of the indices. We apply this analysis
to the selected rows made available by the preceding hierarchy.
Woodbury maps $(\mathbf{B}+\mathbf{E})^{-1}$ to
$(\mathbf{B}'+\mathbf{E})^{-1}$. An inverse-rounding error $\mathbf{F}$ with
$\|\mathbf{F}\|_F\le\delta$ produces
$(\mathbf{B}'+\mathbf{E}+\mathbf{E}')^{-1}$, where
$\|\mathbf{E}'\|_F\le2\|\mathbf{B}'+\mathbf{E}\|_2^2\delta
\le\mathcal K^{O(1)}\delta$,
provided $\delta\|\mathbf{B}'+\mathbf{E}\|_2\le1/2$.
Polynomial factor bounds permit local tolerances with
$\|\mathbf{E}'\|_F\le\eta$. Complete the selected rows virtually with the
exact uncomputed rows of the same Woodbury expression. This full inverse
has the same Frobenius error bound, and every child uses rows of this
common parent. Induction through at most $L+1$ levels from rebuild to
answer gives}
\begin{equation}\label{eq:dynamic-backward-hierarchy}
\begin{aligned}
 \|\mathbf{E}_\ell\|_F\le(L-\ell+1)\eta,\quad
 \|(\mathbf{A}+\mathbf{E}_0)^{-1}-\mathbf{A}^{-1}\|_F
 \le 2\|\mathbf{A}^{-1}\|_2^2\|\mathbf{E}_0\|_F
 \le2\mathcal K^2(L+1)\eta.
\end{aligned}
\end{equation}
Choose $\eta\le\epsilon/[8\mathcal K^2(L+1)]$. Then, stable products and inversion need
only polylogarithmic guard bits beyond $\log(\mathcal K/\epsilon)$. \cite{ABGZ2024} uses
$\tO(\log(\mathcal K/\epsilon))$ bit operations per arithmetic operation.
\looseness=-1
\end{proof}

\section{Finite-precision proofs}
\label{app:precision}

Use the global caps $T_*,V$ and local error $\epsilon$ from
Section~\ref{sec:precision}.

\subsection{Computing the anchor subspace}
\label{app:precision-penalty}

\begin{lemma}
\label{lem:penalty-bound}
Let $\mathbf{S}=\mathbf{S}^\top\in\R^{s\times s}$, let $\mathbf{H}$ have at most $s$ normalized rows,
and suppose a $k$-dimensional subspace $E\subseteq\ker \mathbf{H}$ satisfies
$h^\top \mathbf{S}h\le0$ for every $h\in E$.
Assume $\norm {\mathbf{S}}_2\le W$. For $0<\epsilon_a,\delta_a<1/10$, one can compute
a dyadic $s\times k$ matrix $\mathbf{Z}$ satisfying $\norm{\mathbf{Z}^\top \mathbf{Z}-\mathbf{I}}_2\le\epsilon_a$,
$c^\top \mathbf{Z}^\top \mathbf{S}\mathbf{Z}c\le\epsilon_a\norm{\mathbf{Z}c}_2^2$, and
$\norm{\mathbf{H}\mathbf{Z}c}_2\le2\delta_a\norm{\mathbf{Z}c}_2$ in {expected} $\tO(s^\omega)$ bit operations whenever
$W,\epsilon_a^{-1},\delta_a^{-1}\le n^{O(1)}$.
\end{lemma}

\begin{proof}
Set $\lambda=16(W+1)/\delta_a^2$ and $\mathbf{K}=\mathbf{S}+\lambda \mathbf{H}^\top \mathbf{H}$.
Since $h^\top \mathbf{K}h\le0$ on $E$, min-max gives $\lambda_k(\mathbf{K})\le0$.
Compute $\mathbf{K}=\mathbf{U}\mathbf{D}\mathbf{U}^\top+\mathbf{E}_K$ with $\mathbf{D}$ increasing,
$\|\mathbf{E}_K\|_2\le\epsilon_a/100$, and
$\|\mathbf{U}^\top \mathbf{U}-\mathbf{I}\|_2\le\epsilon_a/[100(1+\|\mathbf{K}\|_2)]$.
Round the first $k$ columns of $\mathbf{U}$ with operator error at most
$\epsilon_a/[1000(1+\|\mathbf{K}\|_2)]$ to obtain $\mathbf{Z}$.
For analysis only, set $\eta_U=\epsilon_a/[100(1+\|\mathbf{K}\|_2)]$
and take the orthogonal polar factor
$\mathbf{Q}=\mathbf{U}(\mathbf{U}^\top\mathbf{U})^{-1/2}$.
Since $\|\mathbf{Q}-\mathbf{U}\|_2\le\eta_U$ and
$\|\mathbf{D}\|_2\le2(1+\|\mathbf{K}\|_2)$,
$\|\mathbf{K}-\mathbf{Q}\mathbf{D}\mathbf{Q}^\top\|_2\le\epsilon_a/10$.
Weyl's inequality gives $D_{kk}\le\epsilon_a/10$, so
$\mathbf{Q}_k^\top\mathbf{K}\mathbf{Q}_k\preceq(\epsilon_a/5)\mathbf{I}$.
The rounding bound gives $\|\mathbf{Z}-\mathbf{Q}_k\|_2\le2\eta_U$,
hence $\|\mathbf{Z}^\top\mathbf{Z}-\mathbf{I}\|_2\le\epsilon_a$ and
$c^\top\mathbf{Z}^\top\mathbf{K}\mathbf{Z}c
\le\tfrac12\epsilon_a\|\mathbf{Z}c\|_2^2$.
For $h=\mathbf{Z}c$, the remaining conclusions follow from\looseness=-1
\begin{align*}
 h^\top\mathbf{S}h&\le h^\top\mathbf{K}h\le\epsilon_a\|h\|_2^2,\qquad
 \lambda\|\mathbf{H}h\|_2^2
 \le(W+\epsilon_a)\|h\|_2^2\le\lambda\delta_a^2\|h\|_2^2.
\end{align*}
Assembly errors $e_S=\|\widehat{\mathbf{S}}-\mathbf{S}\|_2$ and
$e_H=\|\widehat{\mathbf{H}}-\mathbf{H}\|_2$ perturb $\mathbf{K}$ by at most
$e_S+\lambda(2\|\mathbf{H}\|_2e_H+e_H^2)$.
Choose $e_S\le\epsilon_a/1000$ and
$e_H\le\epsilon_a/[1000\lambda(1+\|\mathbf{H}\|_2)]$, and include these
errors in the decomposition residual. Since $\|\mathbf{K}\|_2\le W+\lambda s$,
all norms and inverse tolerances need $O(\log n)$ bits.

{We use a real symmetric version of Shah's gap-free diagonalization
\cite{Shah2024}, with the following modification to its Gaussian deflation.
For an $r\times r$ matrix $\mathbf{B}$ of independent real standard Gaussians,
$\sigma_{\min}(\mathbf{B})\le x$ implies that some column is within
$\sqrt r\,x$ of the span of the others: choose a unit minimizing vector and
a coordinate of magnitude at least $r^{-1/2}$. Conditional on the other
columns, this distance has the law $|N(0,1)|$, hence
$\Pr[\sigma_{\min}(\mathbf{B})\le x]\le Cr^{3/2}x$.
For a real matrix $\mathbf A$ of rank $r\ge1$ with right singular basis $\mathbf V$,
orthogonal invariance makes $\mathbf{V}^\top\mathbf{G}_{:,1:r}$ such a matrix.
If $\mathbf{G}$ has size $d\times d$, set
$x=\rho_f/(4Cr^{3/2})$ and
$L_G=\sqrt{2\log(8d^2/\rho_f)}$.
With failure probability at most $\rho_f/2$,
$\sigma_{\min}(\mathbf{V}^\top\mathbf{G}_{:,1:r})>x$
and $\|\mathbf{G}\|_2\le dL_G$.
Approximate the Gaussian matrix to operator error at most $c\beta x$;
on this event $\|\mathbf{G}\|_2\ge x$, so this is sufficiently small for
the deflation estimate, for a universal $c>0$.
The deterministic deflation estimate in \cite[Section~3]{Shah2024} then gives
subspace error at most $\eta_f$ when its input error satisfies}
\[ \beta\le
 \frac{\eta_f\sigma_r(\mathbf{A})x}
 {12(\sigma_1(\mathbf{A})+2)dL_G}.\]
{For spectral projectors, both singular values in this expression
equal one. Allocate $\rho_f=1/(16s)$ to each of the at most $2(s-1)$
deflations, and choose the real spectral shifts with total failure probability
at most $1/8$. The recursive error allowances are inverse polynomial,
so the smaller $\beta$ requires only $O(\log n)$ additional bits and
polylogarithmic factors in the iteration count. Use real matrix products
and real QR \cite{DDH2007}, and symmetrize every intended symmetric
intermediate matrix. Symmetrization preserves operator error before
rounding; its rounding error has polynomial amplification. Real shifts,
polynomial matrix-sign iterations, transposes, and recursive products
therefore give real $\mathbf{U}$ and $\mathbf{D}$.}

{Accept only when Frobenius upper bounds for the stored residual
and orthogonality defect, including assembly error, imply the required bounds.
Requesting $\sqrt s$ smaller operator errors gives verification slack;
stable products \cite{DDH2007} and dyadic summation cost $\tO(s^\omega)$.
Cap each attempt's arithmetic at its successful-run bound and Gaussian digit
generation at a constant multiple of its mean cost. Markov's inequality bounds
the latter failure by $1/8$, so independent attempts succeed with probability
at least $1/2$. Repetition gives expected $\tO(s^\omega)$ work.
With at most $C_A\log(2n)$ attempts, sufficiently large $C_A$ makes exhaustion
probability at most $n^{-A}$ for any fixed $A>0$. Abort the trial on exhaustion;
verification prevents an incorrect basis from being used, including for small $s$.}
\end{proof}

Apply the lemma to $\mathbf{S}=\mathbf{P}-(1-\rho)\mathbf{D}$ with normalized
nonzero large-row and protected or heavy-row normals. The grid and spread condition
give $|x_j|\le1-\Delta_x$ and $|b_{ij}|/r_i\le1/8$, so
\begin{align*}
 |g_{ij}|=\frac{|b_{ij}|}{r_i}
 \left|\sigma-8\frac{b_{ij}}{r_i}x_j\right|
 \ge\Delta_x\frac{|b_{ij}|}{r_i},
 \quad \norm{g_i}_2\ge\Delta_x\frac{\sqrt{q_i}}{r_i}>\Delta_x/\gamma.
\end{align*}
Large-row normals have norm greater than $R_{\max}$, so normalization
needs $O(\log(\Delta_x^{-1})+\log n)$ bits. Using $t_i,\|\mathbf{R}\|_2\le V$,
$\|g_i\|_2\le2$, $\|z_i\|_1\le1$, and $\sum_i z_{ij}\le\Lambda$ in
\eqref{eq:batch-curvature} gives
$\|\mathbf{S}\|_2\le C(1+\Lambda)^2n^2(1+V)^3=:W$.
For $\mathbf{A}_b=b\mathbf{I}-\mathbf{M}$ and a perturbation $\mathbf{E}_b$ with
$V\norm{\mathbf{E}_b}_2\le1/2$, the inverse identity gives
$\norm{(\mathbf{A}_b+\mathbf{E}_b)^{-1}-\mathbf{R}}_2\le 2V^2\norm{\mathbf{E}_b}_2$ and $\norm{(\mathbf{A}_b+\mathbf{E}_b)^{-1}\ones-v}_2
 \le2\sqrt{s}\,V^2\norm{\mathbf{E}_b}_2$.
Gradient normalization amplifies errors by at most $C\gamma/\Delta_x$.
The at most $2n^2$ row terms, $r_i^{-1}\le q_*^{-1/2}$, and the bounds on
$\lambda$ and inverses give polynomial amplification in
$n,V,q_*^{-1},\Delta_x^{-1},\delta_a^{-1}$. Divide scalar tolerances by
this factor before assembling the anchor matrices. The sampler then
perturbs $\mathbf{Z}$ by independent Gaussians and rounds it.
For $\widehat {\mathbf{Z}}=\mathbf{Z}+\mathbf{E}_Z$ with $\|\mathbf{E}_Z\|_2\le\eta$,\looseness=-1
\begin{align*}
 \norm{\widehat {\mathbf{Z}}^\top \mathbf{S}\widehat {\mathbf{Z}}-\mathbf{Z}^\top \mathbf{S}\mathbf{Z}}_2
 &\le 2\norm {\mathbf{S}}_2\norm {\mathbf{Z}}_2\eta+\norm {\mathbf{S}}_2\eta^2,
 \\
 \norm{\mathbf{H}\widehat {\mathbf{Z}}-\mathbf{H}\mathbf{Z}}_2&\le\norm {\mathbf{H}}_2\eta.
\end{align*}
Choose $\eta$ polynomially smaller than
$\min\{\epsilon_a/(W+1),\delta_a/(\|\mathbf{H}\|_2+1)\}$.
The smoothed basis satisfies twice the original tolerances.
Fixed dyadic grids keep summation growth to $O(\log n)$ bits and prevent
denominator growth. The scale rule in Appendix~\ref{app:rows} and the
fixed-degree weights also need $O(\log n)$ bits. Exact clipping needs
$O(p)$ temporary bits for $p$-bit coordinates and preserves attained endpoints.

\subsection{Transcript coupling}
\label{app:precision-rounding}

\begin{proof}[Proof of Proposition~\ref{prop:finite-implementation}\textup{(i)}]
For the quantizer $q_{\Delta,U}(a)=\Delta\lfloor a/\Delta+U\rfloor$ of
Section~\ref{sec:precision}, $|q_{\Delta,U}(a)-a|\le\Delta$ and, for every
fixed $a$ and $e>0$,
\begin{equation}
\label{eq:dither-boundary}
 \Pr_U\bigl[\operatorname{dist}(a+\Delta U,\Delta\mathbb Z)
                     \le e\bigr]
 \le\frac{2e}{\Delta}+\frac2M.
\end{equation}
{Use a reference with continuous Gaussians, exact normalized outputs,
and the same dithers, endpoint signs, dyadic coloring rules, and shared dyadic
anchor bases. Include smoothing approximation and rounding errors in
$\mathbf E_{\mathrm{round}}$ in \eqref{eq:sampler-smoothed-basis}.
Public geometric, dependence, work-cap, or dither-boundary conditions may stop it.
Its history excludes hidden conditioning and magnitude-guard outcomes.
Apply the conditional sampler estimates in this history. For success, intersect
with the unconditional event $G$ of successful reference execution and polynomial
Gaussian and matrix bounds with guard slack. No estimate conditions on $G$.}

{Approximate the reference Gaussian coefficients within $\delta_G$
using the same lazy random-bit streams. On $G$, until the first stored-output
disagreement, constraint rows agree, so}
\begin{align*}
 \norm{\widehat {\mathbf{Q}}_t-\mathbf{Q}_t}_2\le n^{C_1}\delta_G,\quad
 \norm{\widehat {\mathbf{Q}}_t^{-1}-\mathbf{Q}_t^{-1}}_2
 \le
 \frac{\norm{\mathbf{Q}_t^{-1}}_2^2\norm{\widehat {\mathbf{Q}}_t-\mathbf{Q}_t}_2}
 {1-\norm{\mathbf{Q}_t^{-1}}_2\norm{\widehat {\mathbf{Q}}_t-\mathbf{Q}_t}_2}
 \le 2n^{2C_Q+C_1}\delta_G,
\end{align*}
provided $n^{C_Q+C_1}\delta_G\le1/2$.
{Unconditional Gaussian tails bound magnitude-guard failures;
they do not truncate the reference law. On $G$, near-isometry gives
$(2\|\mathbf{Q}_t\|_2)^{-1}\le\|z\|_2\le2\|\mathbf{Q}_t^{-1}\|_2$,
making normalization polynomially Lipschitz. Taking $\delta_G$ and the dynamic-inverse error polynomially
smaller than $e$ bounds every output error by $e/2$.
For an explicit batch, compute the orthonormal basis by stable Gram-Schmidt
or QR on the near-isometric $\mathbf Z$. The reference uses the exact
orthonormalization of the shared dyadic basis. Since
$\Pr[\|g\|_2\le t]\le\Pr[|g_1|\le t]\le\sqrt{2/\pi}\,t$,
increasing $C_g$ makes all explicit normalization denominators inverse polynomial
with the required probability. The same polynomial error control then applies.
The perturbation estimate, with
the bounded output blocks $\mathbf{Z},\mathbf{Z}^\top\mathbf{Z}$, gives
polynomial norms for the rounded augmented matrix and its inverse,
so the finite guards pass by their slack.}

{Let $D$ be the event that some executed reference scalar $a$ satisfies
$\operatorname{dist}(a+\Delta U,\Delta\mathbb Z)\le e$.
Each new dither is independent of this exact scalar conditional on the preceding
transcript, so \eqref{eq:dither-boundary} and the choices in Section~\ref{sec:precision} give}
\begin{align*}
 \Pr[D]&\le N_{\mathrm{out}}\left(\frac{2e}{\Delta}+\frac2M\right)
 \le\frac12 n^{-C_f}.
\end{align*}
{On $G\cap D^c$, the finite scalar differs from the exact scalar by at most
$e/2$, so its boundary check passes and both processes store the same dyadic value.
Induction gives identical colorings, newly fixed coordinates, constraint rows,
insertion choices, and public stopping decisions. Thus the finite trial succeeds
whenever the reference succeeds on this event, and}
\begin{align*}
 \Pr[\text{finite trial succeeds}]&\ge\Pr[G\cap D^c]
 \ge\Pr[G]-\Pr[D].\qedhere
\end{align*}
\end{proof}

The physical normal associated with $\widehat c$ obeys
$\widetilde g=\mathbf{Z}(\mathbf{Z}^\top \mathbf{Z})^{-1}\widehat c$ and
$\|\widetilde g-h\|_2\le\|\mathbf{Z}(\mathbf{Z}^\top \mathbf{Z})^{-1}\|_2\|\widehat c-c\|_2
 \le2\sqrt k\,\Delta$.
Projection into the feasible space gives angular error
$O(\sqrt k\Delta)$. Masking fixed coordinates cannot increase the error,
since the exact direction vanishes there.
{The comparison above transfers the reference success bounds to the finite
algorithm with fixed precision and an operation cap.} Karney's sampler
\cite{Karney2016} generates Gaussian digits in expected time linear in their number.
Choose tolerances giving Euclidean error at most $\epsilon$ per
contribution after amplification.\looseness=-1
 
\subsection{Propagation of numerical errors}
\label{app:precision-robustness}

\begin{proof}[Proof of Proposition~\ref{prop:finite-implementation}\textup{(ii)}]
For a protected signed row, the hard constraint and local errors give
$|g_i(x_0)^\top(x_{t+1}-x_t)|\le C\epsilon$. Use the quadratic identity
\eqref{eq:barrier-quadratic} with $y=x-x_0$ and sum at most $b_{\max}$ exposures:
\begin{align}
 g_i(x_0)^\top y&=\sum_tg_i(x_0)^\top(x_{t+1}-x_t)
 \le Cb_{\max}\epsilon.
 \label{eq:finite-fixed-identity}
\end{align}
The fixed weight $\nu(u_i(x_0)+\mu_{\mathrm{batch}})$, where
$\mu_{\mathrm{batch}}=Cb_{\max}\epsilon$, therefore dominates the actual
weight throughout the batch without row scans.
For $u\le0$, $\nu'(u)/\nu(u)\le1/2$. Small error allowances make buffering
change weights by a relative factor at most
$e^{\mu_{\mathrm{batch}}/2}-1\le\mu_{\mathrm{batch}}$. The resulting
nonnegative matrix increment $\Delta\mathbf{M}$ satisfies
$0\preceq\Delta \mathbf{M}\preceq\mu_{\mathrm{batch}}\mathbf{M}$ and
$\norm{\Delta \mathbf{M}}_2\le2\mu_{\mathrm{batch}}$.
For $\|\mathbf{R}\|_2\|\Delta \mathbf{M}\|_2\le1/2$,
\begin{align*}
 \bigl\|(\mathbf{I}-\mathbf{R}^{1/2}\Delta \mathbf{M} \mathbf{R}^{1/2})^{-1}\bigr\|_2\le2,\quad
 F(\mathbf{M}+\Delta \mathbf{M})-F(\mathbf{M})
 =v^\top\Delta \mathbf{M}(\mathbf{I}-\mathbf{R}\Delta \mathbf{M})^{-1}v\le C\mu_{\mathrm{batch}}V^2.
\end{align*}
Before buffering, approximate $t_i$ to error at most $\Theta/8$ and fix
rows with $\widehat t_i\ge3\Theta/4$. Selected rows have
$t_i\ge5\Theta/8\ge\Theta/2$, preserving the count, while unselected
rows have $t_i<7\Theta/8$. The entrywise resolvent identity and
$\Delta \mathbf{M}v\le\mu_{\mathrm{batch}}\mathbf{M}v\le2\mu_{\mathrm{batch}}V\ones$
bound the resolvent vector's growth by $(1-2\mu_{\mathrm{batch}}V)^{-1}$.
Choose the error allowance so this factor is at most $8/7$. The remaining
light rows then satisfy $t_i\le\Theta$, as curvature transport requires. For potential drift, \eqref{eq:curv-additive-stability} bounds the Hessian on the anchor span by
$2\epsilon_a/\rho$ times the identity throughout the enlarged window.
Directions within $\epsilon$ of this span add at most
$C\norm{\nabla^2F}_2\epsilon$. The resolvent formulas bound
$\norm{\nabla F}_2\le C(1+V)^2$ and
$\norm{\nabla^2F}_2\le C(1+V)^3$.
Since $\nu^{(r)}(u)\le C_r\nu(u)$ and $\|g_i\|_2\le2$, the first two
matrix derivatives have constant bounds. Each inverse derivative adds at
most one factor $\|\mathbf{R}\|_2\le V$.

Condition on the direction before drawing the fair sign. Symmetric Taylor
expansion then gives
$\E_\sigma F(x+\sigma th)-F(x)\le C\tau^2(1+V)^3\epsilon$,
where $\epsilon$ includes anchor error. Endpoint rounding adds at most
$C(1+V)^2\epsilon$. Summing exposure and buffer costs over total batch
capacity $O(T_*+n)$ gives drift at most
$CT_*(1+V)^3\epsilon$. Adding the remaining error allowance restores a
nonnegative supermartingale and the potential bound's constant success
probability.
The condition $u_i\le0$ has slack. Newly medium rows begin far below $-1$.
Unprotected rows start each batch below $-1$, and the enlarged
window raises their barriers by at most $1/4$. Across consecutive protected batches,
\eqref{eq:finite-fixed-identity} bounds the cumulative increase. Rows
leaving protection are again below $-1$, so $CT_*\epsilon<1/100$ keeps
each protected episode below $-1/2$. Downward row reductions preserve
this margin. Buffered barriers stay negative, avoiding decisions at $u=0$.
For large rows, normalized hard constraints limit drift to
$C\|b_i\|_2\epsilon\le C\sqrt n\epsilon$ per exposure, so
$|d_i|\le C\sqrt n\,T_*\epsilon$ before the row becomes medium.
Choose $\delta_d=C\sqrt n\,T_*\epsilon\le10^{-3}$.
The scalar comparison proving \eqref{eq:birth-deposit}, with $H_0$
replaced by $H_0-\delta_d$, remains increasing in $r\le R_{\max}$.
Its largest deposit ratio is therefore at most
$e^{\delta_d/(2R_{\max})}(1000/1009)<1$.
Deposit replacement still decreases the majorant in Loewner order,
and negative barriers permit the same subsequent scale reductions.
Let $\mathbf{G}$ contain all unit query directions in a batch, including
the final query if it has no insertion, and let $\mathbf{H}$ contain
their scaled, rounded proposals. The pairwise argument of
\eqref{eq:sampler-self-gram} gives
$\|\mathbf{G}\mathbf{G}^\top-\mathbf{I}\|_2\le Cb_{\max}\epsilon$.

Each exact direction vanishes on the coordinates fixed before its query,
so masking its rounded proposal changes only its error. Summing these
row errors gives $\|\mathbf{H}-3\mathbf{G}/4\|_F\le C\sqrt{b_{\max}}\epsilon$,
even though the masks vary. Thus $\|\mathbf{H}\|_2^2\le2$, preserving the
window proof's squared-projection bound.
Coordinate rounding adds at most $T_*\epsilon$
to any unit linear functional.
Expose each direction and clipping length before drawing its independent
fair endpoint sign, and stop before a proposal violates the Gram or cap
bounds. The sign-only argument of Proposition~\ref{prop:walk-success} then
gives \eqref{eq:walk-retirement-tail}, with changed absolute constants.
Coordinate rounding adds drift at most $T_*\epsilon\sqrt{q_*}$.
For sufficiently small $c_*$, $q_*=c_*L^{-4}$ preserves the retirement
failure allowance over fewer than $n^2$ rows.
Endpoint rounding changes energy by at most $C\sqrt n\epsilon$ per
exposure. The norm and clipping bounds persist, so
Lemma~\ref{lem:walk-query-count} applies with energy allowance
$N+C\sqrt nT_*\epsilon$.
\end{proof}

\end{document}